\documentclass[10pt]{amsart}
\usepackage{amsmath,amssymb,subfigure,floatflt}
\usepackage{epsfig}
\usepackage{calrsfs}
\newtheorem{theorem}{Theorem}[section]
\newtheorem{lemma}[theorem]{Lemma}
\newtheorem{corollary}[theorem]{Corollary}
\newtheorem{proposition}[theorem]{Proposition}
\newtheorem{definition}{Definition}[section]
\newtheorem{example}[theorem]{Example}

\allowdisplaybreaks
\numberwithin{equation}{section}

\newcommand{\nc}{\newcommand}

\nc{\pq}{\left[^{\boldsymbol{\delta}}_{\boldsymbol{\epsilon}}\right]}
\nc{\mpq}{\left[^{-\boldsymbol{\delta}}_{-\boldsymbol{\epsilon}}\right]}
\nc{\pqt}{\left[^{\tilde{\boldsymbol{\delta}}}_{\tilde{\boldsymbol{\varepsilon}}}
\right]}
\nc{\de}{\left[^{\delta}_{\varepsilon}\right]}

\DeclareMathOperator{\Der}{Der}

\begin{document}

..
\dedicatory{To Professor Emma Previato on the occasion of her 65th birthday}

\title[Multi-variable sigma-functions: old and new results]{Multi-variable sigma-functions: old and new results}
\author{V.M.Buchstaber}
\address{Steklov Mathematical Institute, Moscow}
\email{buchstab@mi-ras.ru}
\author{V.Z. Enolski}
\address{National University of Kyiv-Mohyla Academy}
\email{venolski@gmail.com,enolsky@ukma.edu.ua}
\author{D.V.Leykin}
\address{Institute of Magnetism NASU}
\email{dmitry.leykin@gmail.com}

\maketitle

\begin{abstract}
  We consider multi-variable sigma function of a genus $g$
  hyperelliptic curve as a function of two group of variables -
  jacobian variables and parameters of the curve.  In the theta-functional
  representation of sigma-function, the second group arises as periods
  of first and second kind differentials of the curve. We develop
  representation of periods in terms of theta-constants. For the first
  kind periods, generalizations of Rosenhain type formulae are obtained,
  whilst for the second kind periods theta-constant expressions are
  presented which are explicitly related to
  the fixed co-homology basis.\\
  We describe a method of constructing differentiation operators for
  hyperelliptic analogues of $\zeta$- and $\wp$-functions on the
  parameters of the hyperelliptic curve.  To demonstrate this method,
  we gave the detailed construction of these operators in the cases of
  genus 1 and 2.
\end{abstract}

\section{Introduction}
Our note belongs to an area in which Emma Previato took active part in
the development. Since the time of first publication of the present
authors \cite{bel97} she has inspired them, and given them a lot of
suggestions and advices.

The area under consideration is the construction of Abelian functions
in terms of multi-variable $\sigma$-functions.  Similarly to the
Weierstrass elliptic function, the multi-variable sigma keeps the same
main property - it remains form-invariant at the action of the
symplectic group.  Abelian functions appear as logarithmic derivatives
like the $\zeta, \wp$-functions of the Weierstrass theory and similarly to
the standard theta-functional approach which lead to the Krichever formula
for KP solutions, \cite{kr1977}.  But a fundamental difference between
sigma and theta-functional theories is the following.  They are both
constructed by the curve and given as series in Jacobian variables,
but in the first case the expansion is purely algebraic with respect
to the model of the curve since its coefficients are polynomials in
parameters of the defining curve equation, whilst in the second case
coefficients are transcendental, being built in terms of Riemann
matrix periods which are complete Abelian integrals.

In many publications, in particular see
\cite{oni05,eemop07,bel12,ny12,ehkklp12,an13,bef13,mp14,nak16,nak18}
and references therein, it was demonstrated that multi-variable
$\wp$-functions represent a language which is very suitable to speak
about completely integrable systems of KP type. In particular, very
recently, using the fact that sigma is an entire function in the parameters
of the curve (in contrast with theta-function), families of degenerate
solutions for solitonic equations were obtained \cite{bl18,ben18}.  In
recent papers \cite{bm17} and \cite{bm18}, an algebraic construction
of a wide class of polynomial Hamiltonian integrable systems was
given, and those of them whose solutions are given by hyperelliptic
$\wp$-functions were indicated.

Revival of interest in multi-variable $\sigma$-functions is in many
respects guided by H.Baker's exposition of the theory of Abelian
functions, which go back to K.Weierstrass and F.Klein and is well
documented and developed in his remarkable monographs
\cite{bak998,bak903}. The heart of his exposition is the representation of
fundamental bi-differential of the hyperelliptic curve in algebraic form,
in contrast with the representation as a double differential
of Riemann theta-function developed by Fay in the monograph,
\cite{fay973}. Recent investigations have demonstrated that Baker's
approach can be extended beyond hyperelliptic curves to wider classes of
algebraic curves.

The multi-variable $\sigma$-function of algebraic curve $\mathcal{C}$
is known to be represented in terms of $\theta$-function of the curve
as a function of two groups of variables - the Jacobian of the curve,
$\mathrm{Jac}(\mathcal{C})$ and the Riemann matrix $\tau$.  In the
vast amount of recent publications, properties of the
$\sigma$-function as a function of the first group of variables were
discussed, whilst modular part of variables and relevant objects like
$\theta$-constant representations of complete Abelian integrals are
considered separately. In this paper we deals with $\sigma$ as a
function over both group of variables.

Due to the pure modular part of the $\sigma$-variables, we consider
problem of expression of complete integrals of first and second kind
in terms of theta-constants.  A revival of interest in this
classically known material accords to many recent publications
reconsidering such problems such as the Schottky problem \cite{fgs17},
the Thomae \cite{eg06,ef08,ekz18} and Weber \cite{nr17} formulae, and the
theory of invariants and its applications \cite{ksv05,ef16}.

The paper is organized as follows. In the Section 2 we consider the
hyperelliptic genus $g$ curve and the complete co-homology basis of
$2g$ meromorphic differentials, with $g$ of them chosen as holomorphic
ones.  We discuss expressions for periods of these differentials in
terms of $\theta$-constants with half integer
characteristics. Theta-constants representations of periods of
holomorphic integrals is known from the Rosenhain memoir
\cite{ros851}, where the case of genus two was elaborated.  We discus
this case and generalize the Rosenhain expressions to higher genera
hyperelliptic curves. Theta constant representations of second kind
periods is known after F.Klein \cite{klein888}, who presented closed
formula in terms of derivated even theta-constants for the
non-hyperelliptic genus three curve. We re-derive this formula for
higher genera hyperelliptic curves.

Section 3 is devoted to a classically known problem, which was
resolved in the case of elliptic curves by Frobenius and Stickelberger
\cite{fs882}. The general method of the solution of this problem for a
wide class, the so-called $(n, s)$ - curves, has been developed in
\cite{bl08} and represents extensions of the Weierstrass's method for
the derivation of system of differential equations defining the
sigma-function. All stages of derivation are given in details and the
main result is that the sigma-function is completely defined as the
solution of a system of heat conductivity equations in a nonholonomic
frame.  We also consider there another widely known problem - the
description of the dependence of the solutions on initial data. This
problem is formulated as a description of the dependence of the
integrals of motion, which levels are given as half of the curve
parameters, from the remaining half of the parameters.  The
differential formulae obtained permit us to present effective
solutions of this problem for Abelian functions of the hyperelliptic
curve.  It is noteworthy that because integrals of motion can be
expressed in terms of second kind periods in this place, the results
of Section 2 are required. The results obtained in Section 3 are
exemplified in details by curves of genera one and two.  All
consideration is based on explicit uniformization of the space of
universal bundles of the hyperelliptic Jacobian.

\section{Modular representation of periods of hyperelliptic co-homologies}
Modular invariance of the Weierstrass elliptic $\sigma$-function, $\sigma=\sigma(u;g_2,g_3)$
follows from its defining in in \cite{wei885} in terms of recursive series in terms of
variables $(u,g_2,g_3)$. Alternatively $\sigma$-function can be represented
in terms of Jacobi $\theta$-function and its modular invariance follows from transformation
properties of $\theta$-functions. The last representation involves complete elliptic integrals
of first and second kind and their representations in terms of $\theta$-constants are classically known.
In this section we are studying generalizations of these representations to hyperelliptic curves
of higher genera realized in the form
\begin{align} y^2 = P_{2g+1}(x)=   (x-e_1)
\cdots (x-e_{2g+1}) \label{HCurve} \end{align}
Here  $P_{2g+1}(x)$ - monic polynomial of degree $2g+1$, $e_i\in \mathbb{C}$ - branch points  and the curve supposed
to be non-degenerate, i.e. $e_i\neq e_j$.

\subsection{\underline{Problems and methods}}
Representations of complete elliptic integrals of first and second kind in terms of Jacobi $\theta$-constants are classically known. In particular, if elliptic curve is given in Legendre form\footnote{Here and below we punctually follows notations of elliptic functions theory fixed in \cite{be955} }
\begin{equation}
y^2=(1-x^2)(1- k^2x^2)
\end{equation}
where $k$ is Jacobian modulus, then complete elliptic integrals of the first kind $K=K(k)$
is represented as
\begin{align}    K=  \int_{0}^1 \frac{\mathrm{d} x}{ \sqrt{ (1-x^2)(1-k^2 x^2) }}= \frac{\pi}{2} \vartheta_3^2(0;\tau), \label{FirstKind}
 \end{align}
and  $\vartheta_3=\vartheta_3(0|\tau)$ and $\tau=\imath \frac{K'}{K}$, $K'=K(k'),  k^2+{k'}^2=1$.

Further, for elliptic curve realized as Weierstrass cubic
\begin{equation} y^2= 4 x^3-g_2 x - g_3=  4 (x-e_1)(x-e_2)(x-e_3) \label{WCubic}   \end{equation}
recall standard notations for periods of first and second kind elliptic integrals
\begin{align}  \begin{split}
  2\omega &= \oint_{\mathfrak{a}} \frac{\mathrm{d}x}{y}, \quad  2\eta = -\oint_{\mathfrak{a}} \frac{x\mathrm{d}x}{y}\\
  2\omega' &= \oint_{\mathfrak{b}} \frac{\mathrm{d}x}{y}, \quad  2\eta' = -\oint_{\mathfrak{b}} \frac{x\mathrm{d}x}{y}
\end{split} \hskip1.5cm    \tau = \frac{\omega'}{\omega}
\end{align}
and Legendre relation for them
 \begin{equation}  \quad\omega\eta'-\eta\omega' = - \frac{\imath \pi}{2}  \label{LegendreRel}    \end{equation}
Then the following Weierstrass relation is valid
\begin{align} \eta=-  \frac{1}{12\omega} \left(\frac{\vartheta_2''(0)}{\vartheta_2(0)}
+ \frac{\vartheta_3''(0)}{\vartheta_3(0)}+ \frac{\vartheta_4''(0)}{\vartheta_4(0)}\right)
\label{SecondKind}
\end{align}
In this section we are discussing generalization of these relations to higher genera hyperelliptic curves
realized as in (\ref{HCurve}).

\subsection{\underline{Definitions and main theorems}}
In this subsection we reproduce  H.Baker \cite{bak907} notations.
Let $\mathcal{C}$ be genus $g$ non-degenerate
hyperelliptic  curve realised as double cover of Riemann sphere,
\begin{align}
y^2= 4 \prod_{j=1}^{2g+1}(x-e_j)\equiv 4x^{2g+1}+ \sum_{i=0}\lambda_i x^i,\quad e_i\neq e_j,\; \lambda_i
\in \mathbb{C}
\label{hyperelliptic}
\end{align}
Let $({\mathfrak{a}};{\mathfrak{b}})
=( \mathfrak{a}_1,\ldots,\mathfrak{a}_g;  \mathfrak{b}_1,\ldots,\mathfrak{b}_g )  $ be canonic homology basis.
Introduce co-homology basis ({\em Baker co-homology basis})
\begin{align} \begin{split} \mathrm{d}{u} (x,y)=(  \mathrm{d}u_1 (x,y), \ldots,
 \mathrm{d}u_g (x,y))^T,   \mathrm{d}{r} (x,y)=(  \mathrm{d}r_1 (x,y),    \ldots,
    \mathrm{d}r_g (x,y))^T\\
\mathrm{d}u_i(x,y)=\frac{x^{i-1}}{y}\mathrm{d}x,\quad \mathrm{d}r_j(x,y)=\sum_{k=j}^{2g+1-j}(k+1-j)
\frac{x^k}{4y}\mathrm{d}x,\quad i,j=1,\ldots,g\end{split}\label{bakerbasis}
\end{align}
satisfying to the generalized Legendre relation,
\begin{align}
\mathfrak{M}^TJ\mathfrak{M}=-\frac{\imath \pi}{2} J,\qquad \mathfrak{M}=\left(\begin{array}{cc} \omega&\omega'\\
                        \eta& \eta' \end{array}\right), \quad J = \left( \begin{array}{cc}  0_g& 1_g \\ -1_g&0_g
\end{array}  \right)
\end{align}
where $g\times g$ period matrices $\omega,\omega', \eta, \eta'$ are defined as
\begin{align}
2\omega=\left( \oint_{\mathfrak{a}_j}\mathrm{d}u_i \right),\quad
2\omega'=\left( \oint_{\mathfrak{b}_j}\mathrm{d}u_i \right), \quad
2\eta=-\left( \oint_{\mathfrak{a}_j}\mathrm{d}r_i \right),\quad
2\eta'=-\left( \oint_{\mathfrak{b}_j}\mathrm{d}r_i \right)
\end{align}
We also denote $\mathrm{d}v=(\mathrm{d}v_1,\ldots,\mathrm{d}v_g )^T =  (2\omega)^{-1} \mathrm{d}u$ vector of 
normalized holomorphic differentials.

Define Riemann matrix $\tau= \omega^{-1}\omega'$ belonging to Siegel half-space $\mathcal{S}_g =\{ \tau^T=\tau,\;  \mathrm{Im} \tau>0 \}  $. Define Jacobi variety of
the curve $\mathrm{Jac}(\mathcal{C})=\mathbb{C}^g/ 1_g\oplus \tau  $. Canonic Riemann $\theta$-function
is defined on $\mathrm{Jac}(\mathcal{C})\times \mathcal{S}_g$ by Fourier series
\begin{equation}
\theta({z};\tau) = \sum_{\mathbb{n} \in \mathbb{Z}^n} \mathrm{e}^{ \imath \pi {n}^T\tau {n}
+ 2\imath \pi {z}^T {n}  }\label{thetacan}
\end{equation}
We will also use $\theta$-functions with half-integer characteristics $[\varepsilon] =
\left[\begin{array}{c} {\varepsilon'}^T\\ {\varepsilon''} \end{array}  \right]$,
$\varepsilon_i', \varepsilon_j'' = 0$ or $1$ defined as
 \begin{equation}
\theta[\varepsilon]({z};\tau) = \sum_{\mathbb{n} \in \mathbb{Z}^n} \mathrm{e}^{ \imath \pi
( {n+\varepsilon'}/2)^T
\tau( {n+\varepsilon'}/2)
+ 2\imath \pi ({z+\varepsilon''}/2)^T( n+\varepsilon'/2)  }\label{thetachar}
\end{equation}
Characteristic is even or odd whenever  $ {\varepsilon'}^T\varepsilon'' = 0 $ (mod 2) or  1 (mod 2) and $\theta[\varepsilon](z;\tau)$ as function of $z$ inherits parity of the characteristic.

Derivatives of $\theta$-functions by arguments $z_i$ will be denoted as
\[ \theta_i[\varepsilon](z;\tau) = \frac{\partial}{\partial z_i}  \theta[\varepsilon](z;\tau),
\quad \theta_{i,j}[\varepsilon](z;\tau) = \frac{\partial^2}{\partial z_i\partial z_j}  \theta[\varepsilon](z;\tau) , \quad\text{etc.}  \]

Fundamental bi-differential $\Omega(P,Q)$ is uniquely definite on
the product $(P,Q)\in \mathcal{C}\times \mathcal{C}$ by following conditions:

{\bf i} $\Omega$ is symmetric, $\Omega(P,Q)=\Omega(Q,P)$

{\bf ii} $\Omega$ is normalized by the condition
\begin{align}
\oint_{\mathfrak{a}_i} \Omega(P,Q)=0, \quad i=1,\ldots,g
\end{align}

{\bf iii} Let $P=(x,y)$ and $Q=(z,w)$ have local coordinates $\xi_1=\xi(P)$, $\xi_2=\xi(Q)$ in the vicinity
of point $R$, $\xi(R)=0$, then $\Omega(P,Q)$ expands to power series as
\begin{equation}
\Omega(P,Q)= \frac{\mathrm{d}\xi_1 \mathrm{d}\xi_2}{( \xi_1-\xi_2)^2 } + \; \text{homorphic 2-form}
\end{equation}

Fundamental bi-differential can be expressed in terms of $\theta$-function \cite{fay973}
\begin{equation}
\Omega(P,Q)= \mathrm{d}_x\mathrm{d}_z\theta\left( \int_{Q}^P\mathrm{d}{v} + {e}\right), \quad P=(x,y), Q=(z,w)
\end{equation}
where $\mathrm{d}v$ is normalized holomorphic differential and ${e}$ any non-singular point of the $\theta$-divisor $(\theta)$,
i.e. $\theta({e})=0$, but not all
$\theta$-derivatives, $ \partial_{z_i}\theta({z})\vert_{{z}={e}}  $, $i=1,\ldots,g$
vanish.

In the case of hyperelliptic curve $\Omega(P,Q)$ can be alternatively constructed as
\begin{equation}
\Omega(P,Q) =\frac12 \frac{\partial}{\partial z} \frac{ y+w }{ y(x-z)} \mathrm{d}x\mathrm{d}z
+ \mathrm{d}{r}(P)^T \mathrm{d}{u} (Q) + 2\mathrm{d}{u}^T (P)\varkappa
\mathrm{d}{u} (Q)\label{omega1}
\end{equation}
where first two terms are given as rational functions of coordinates $P,Q$ and  necessarily symmetric matrix
$\varkappa^T=\varkappa$,  $\varkappa=\eta(2\omega)^{-1}$ is introduced
to satisfy the normalization condition ${\bf ii}$. In shorter form (\ref{omega1}) cab be rewritten as
\begin{equation}
\Omega(P,Q) =\frac{2 yw +F(x,z)}{ 4 (x-z)^2  yw } \mathrm{d}x\mathrm{d}z + 2\mathrm{d}{u}^T (P)\varkappa
\mathrm{d}{u} (Q)\label{omega1}
\end{equation}
where $F(x,z)$ is so-called Kleinian 2-polar, given as
\begin{equation}
F(x,z)=\sum_{k=0}^g x^kz^k  \left( 2\lambda_{2k}+\lambda_{2k+1}(x+z)  \right) \label{polar}
\end{equation}
Recently algebraic representation for $\Omega(P,Q)$ similar to (\ref{omega1}) found in
\cite{suz17}, \cite{eyn18} for wide class on algebraic curves, included $(n,s)$-curves \cite{bel999}.

Main relation lying in the base of the theory is Riemann formula represented third kind Abelian integral
as $\theta$-quotient written in terms of  described above realization of the fundamental differential $\Omega(P,Q)$.

\begin{theorem} (Riemann) Let  $P'=(x',y')$ and $P''=(x'',y'')$ are two arbitrary distinct points
on $\mathcal{C}$ and let $\mathcal{D}'=\{ P_1'+\ldots+P_{g}'\}$ and  $\mathcal{D}''=\{ P_1''+\ldots+P_{g}''\}$
are two non-special divisors of degree $g$. Then the following relation is valid
\begin{align}\begin{split}
&\int_{P''}^{P'} \sum_{j=1}^g\int_{P_j'}^{P_j''} \left\{  \frac{2yy_i+F(x,x_i)}{4(x-x_i)^2}\frac{\mathrm{d}x}{y}\frac{\mathrm{d}x_i}{y_i} + 2 \mathrm{d}{u}(x,y)\varkappa \mathrm{d} {u}(x_i,y_i)   \right\}  \\
&=\mathrm{ln}
\left(\frac{  \theta(\mathcal{A}(P')-\mathcal{A}(\mathcal{D}') +{K}_{\infty} )}
           { \theta(\mathcal{A}(P')-\mathcal{A}(\mathcal{D}'') +{K}_{\infty} ) }       \right)
-\mathrm{ln}
\left(\frac{\theta(\mathcal{A}(P'')-\mathcal{A}(\mathcal{D}') +{K}_{\infty} )}
           {\theta(\mathcal{A}(P'')-\mathcal{A}(\mathcal{D}'') +{K}_{\infty} ) } \right)
\end{split} \label{Riemann1}
\end{align}
where $\mathcal{A}(P)= \int_{\infty}^{P} \mathrm{d}{v}$ Abel map with base point $\infty$,
${K}_{\infty}$ - vector of Riemann constants with bases point $\infty$ which is a half-period.
\end{theorem}

Introduce multi-variable fundamental $\sigma$-function,
\begin{equation}
\sigma({u}) = C \theta[{K}_{\infty}]( (2\omega)^{-1}{u})
\mathrm{e}^{ {u}^T\varkappa{u}   },
\end{equation}
where $[{K}_{\infty}]$ is characteristic of the vector of Riemann constants,
${u}= \int_{\infty}^{P_1} \mathrm{d}{u} + \ldots +  \int_{\infty}^{P_g} \mathrm{d}{u}   $
with non-special divisor $P_1+\ldots+P_g$. The constant $C$  is chosen so that expansion $\sigma(\boldsymbol{u})$
near ${u}\sim 0$ starts with a Schur-Weierstrass polynomial \cite{bel999}.
The whole expression is proved to be invariant under the action of symplectic group $\mathrm{Sp}(2g,\mathbb{Z})$.
Klein-Weierstrass multi-variable $\wp$-functions are introduced as logarithmic derivatives,
\begin{align}
\wp_{i,j}({u})=-\frac{\partial^2}{\partial u_i\partial u_j}, \quad \wp_{i,j,k}
({u})=-\frac{\partial^3}{\partial u_i\partial u_j\partial u_k},\quad \text{etc.} \quad
i,j,k = 1,\ldots g
\end{align}

\begin{corollary} For $r\neq s \in \{1,\ldots, g\}$ the following formula is valid
\begin{equation}
\sum_{i,j=1}^g\wp_{i,j} \left( \sum_{k=1}^g \int_{\infty}^{(x_k,y_k)} \mathrm{d}{u}   \right)x_s^{i-1} x_r^{j-1}
= \frac{F(x_s,x_r)-2y_sy_r}{ 4(x_s-x_r)^2}
\end{equation}
\end{corollary}

\begin{corollary} Jacobi problem of inversion of the Abel map $\mathcal{D} \rightarrow \mathcal{A}(  \mathcal{D})$
with $\mathcal{D}= (x_1,y_1)+ \ldots+ (x_g,y_g)$ is resolved as
\begin{align}\begin{split}
&x^g-\wp_{g,g}({u})x^{g-1} - \ldots -\wp_{g,1}({u})=0\\
&y_k= \wp_{g,g,g}({u})x_k^{g-1}+ \ldots +  \wp_{g,g,1}({u})
, \quad k=1,\ldots,g\end{split}
\end{align}
\end{corollary}

\subsection{\underline{$\wp$-values at non-singular even half-periods}}
In this section we present generalization of Weierstrass formulae
\begin{equation}
\wp(\omega)=e_1,\; \wp(\omega+\omega')=e_2,\; \wp(\omega')=e_3
\end{equation}
to the case of genus $g$ hyperelliptic curve (\ref{hyperelliptic}). To do that introduce partitions
\begin{align}\begin{split} \{1,\ldots, 2g+1\} = \mathcal{I}_0\cup \mathcal{J}_0, \quad   \mathcal{I}_0\cap \mathcal{J}_0=\emptyset\\
\mathcal{I}_0=\{ i_1,\ldots, i_g \}, \quad \mathcal{J}_0=\{ j_1,\ldots, j_{g+1} \} \end{split}
 \end{align}
Then any non-singular even half-period $\Omega_{\mathcal{I}}$   is given as
\begin{align}
\Omega_{\mathcal{I}_0} = \int_{\infty}^{( e_{i_1},0  )}\mathrm{d} {u}
+\ldots+  \int_{\infty}^{( e_{i_g},0  )}\mathrm{d} {u},\quad  \mathcal{I}_0= \{i_1,\ldots, i_g\}\subset \{1,\ldots,2g+1\}
 \end{align}

Denote elementary symmetric functions  $s_n(\mathcal{I}_0)$, $S_{n}(\mathcal{J}_0)$  of order $n$
built on branch points $\{ e_{i_k} \}$,
$i_k\in \mathcal{I}_0$, $\{ e_{j_k} \}$, $j_k\in \mathcal{J}_0$ correspondingly.  In particular,
\begin{align}\begin{split}
s_1(\mathcal{I}_0)&=e_{i_1}+\ldots+ e_{i_g}, \hskip1.65cm S_1(\mathcal{J}_0)=e_{j_1}+\ldots+ e_{j_{g+1}}\\
s_2(\mathcal{I}_0)&=e_{i_1}e_{i_2}+\ldots+e_{i_{g-1}} e_{i_g},
\hskip0.5cm S_2(\mathcal{J}_0)=e_{j_1}e_{j_2}+\ldots+e_{j_{g}} e_{i_{g+1}}\\
&\vdots \hskip5.4cm\vdots\\
s_g(\mathcal{I}_0)&=e_{i_1}\cdots e_{i_g} \hskip 2.62cm
S_{g+1}(\mathcal{J}_0)=e_{j_1}\cdots e_{j_{g+1}}
\end{split}
\end{align}
Because of symmetry, $\wp_{p,q}({\Omega}_{\mathcal{I}_0})=\wp_{q,p}({\Omega}_{\mathcal{I}_0})$ enough to
find these quantities for $ p\leq q \in \{1,\ldots,g\}  $. The following is valid

\begin{proposition}(Conjectural Proposition ) Let even non-singular half-period ${\Omega}_{\mathcal{I}_0}$ associate to a partition
$\mathcal{I}_0\cup \mathcal{J}_0=\{1,\ldots,2g+1\}$. Then for all $k,j\geq k, k,j = 1\ldots,g$
the following formula is valid

\begin{align}\begin{split}
&\wp_{k,j}({\Omega}_{\mathcal{I}_0})\\&=(-1)^{k+j}\sum_{n=1}^k n \left( s_{g-k+n}(\mathcal{I}_0)S_{g-j-n+1}(\mathcal{J}_0)
+s_{g-j-n}(\mathcal{I}_0)S_{g+n-k+1}(\mathcal{J}_0) \right),
 \end{split} \label{formula}
\end{align}
\end{proposition}

\begin{proof}Klein formula written for even non-singular half period ${\Omega}_{\mathcal{I}_0}$ leads to
linear system of equations with respect to Kleinian two-index symbols $ \wp_{i,j}({\Omega}_{\mathcal{I}_0}) $
\begin{align}
\sum_{i=1}^g\sum_{j=1}^g \wp_{i,j}({\Omega}_{\mathcal{I}_0})   e_{i_r}^{i-1}e_{i_s}^{j-1} =\frac{F(e_{i_r}, e_{i_s})}
{4(e_{i_r} - e_{i_s})^2}\quad  \; i_r, i_s \in \mathcal{I}_0\label{Kleineq}
\end{align}

To solve these equation we note that
\begin{equation}
\wp_{k,g}({\Omega}_{\mathcal{I}_0})= (-1)^{k+1}s_k(\mathcal{I}_0), \quad k=1,\ldots,g
\end{equation}
Also note that $F(e_{i_r}, e_{i_s})$ is divisible by $(e_{i_r}- e_{i_s})^2$ and
\begin{equation}
\frac{F(e_{i_r}, e_{i_s})}{4(e_{i_r} - e_{i_s})^2}= e_{i_r}^{g-1} e_{i_s}^{g-1}\mathfrak{S}_1+
e_{i_r}^{g-2} e_{i_s}^{g-2}\mathfrak{S}_2+\ldots + \mathfrak{S}_{2g-1}\label{JIP1}
\end{equation}
where $\mathfrak{G}_k$ are order $k$ elementary symmetric functions of elements $e_i$
$ i\in \{1, \ldots, 2g+1 \} -  \{ i_r,i_s  \}  $

Let us analyse equations (\ref{Kleineq}) for small genera, $g\leq 5$. One can see that plugging
to the equation (\ref{JIP1}) to (\ref{Kleineq}) we get non-homogeneous linear equations
solvable by Kramer rule and the solutions can be presented in the form (\ref{formula}).

Now suppose that (\ref{formula}) is valid for higher $g>5$ were computer power
 is insufficient
to check that by means of computer algebra. But that's possible to check  (\ref{formula}) for arbitrary big genera
numerically giving to branch points $e_i, i=1,\ldots, i=2g+1$ certain numeric values. Many checking confirmed
 (\ref{formula}) .  \end{proof}

\subsection{\underline{Modular representation of $\varkappa$ matrix}}
Quantities  $\wp_{i,j}({\Omega}_{\mathcal{I}_0})$ are expressed in terms of even $\theta$-constants as follows
\begin{align}
 \wp_{i,j}({\Omega}_{\mathcal{I}_0})=-2\varkappa_{i,j}
- \frac{1}{\theta[\varepsilon_{\mathcal{I}_0}]({0})}
\partial_{{U}_i,{U_j} }^2 \theta[\varepsilon_{\mathcal{I}_0}]({0}), \quad \forall \mathcal{I}_0,\; i,j=1,\ldots, g.
\end{align}
Here $[\varepsilon_{\mathcal{I}_0}]$ is characteristic of the vector $\left[{\Omega}_{\mathcal{I}_0}+{K}_{\infty}  \right]$, where
${K}_{\infty}$ is vector of Riemann constants with base point $\infty$ and
$\partial_{{U}}$ is directional derivative along vector $\boldsymbol{U}_i$, that is $i$th column vector of inverse matrix of $\mathfrak{a}$-periods,
$\mathcal{A}^{-1} = ( {U}_1,\ldots, {U}_g )$. The same formula is valid for all possible partitions
$\mathcal{I}_0\cup \mathcal{J}_0$, there are $N_g$ of that, that is number of non-singular even characteristics,
\begin{equation}
N_g= \left( \begin{array}{c}  2g+1\\ g  \end{array}  \right)
\end{equation}
Therefore one can write

\begin{align}
\begin{split}
\varkappa_{i,j} &=\frac{1}{8N_g} \Lambda_{i,j}- \frac{1}{2 N_g}
\sum_{  \text{All even non-singular} \;\;  [\varepsilon]  }  \frac{\partial^2_{{U}_i,{U_j}}
 \theta[\varepsilon_{\mathcal{I}_0}]({0})}{\theta[\varepsilon_{\mathcal{I}_0}]({0})}\\
\end{split}
\end{align}
where
\begin{align}
\Lambda_{i,j}=-4\sum_{  \text{All partitions} \; \mathcal{I}_0  } \wp_{i,j}({\Omega}_{\mathcal{I}_0})
\end{align}
Denote by $\Lambda_g$ symmetric matrix
\begin{equation} \Lambda_g = ( \Lambda_{i,j} )_{i,j=1,\ldots,g} \label{Lambda1}  \end{equation}

\begin{proposition}  Entries $\Lambda_{k,j}$ at $k\leq j$ to the symmetric matrix $\Lambda$ are given by the formula
\begin{align}\begin{split}
\Lambda_{k,j}&= \lambda_{k+j}  \frac{ \left( \begin{array}{c} 2g+1\\g \end{array}  \right)  }
{ \left( \begin{array}{c} 2g+1\\2g+1-k-j \end{array}  \right)  }
\sum_{n=1}^k n \left[  \left( \begin{array}{c} g\\ g-k+n \end{array}  \right) \left( \begin{array}{c} g+1\\ g-j-n+1 \end{array}  \right)  \right. \\&\hskip 5.5cm + \left.\left( \begin{array}{c} g \\g-j-n \end{array}  \right)
  \left( \begin{array}{c} g+1\\ g-k+n+1 \end{array}  \right)  \right]\end{split}\label{Lambda2}
\end{align}
\end{proposition}
\begin{proof}Execute summation in (\ref{formula}) and
find that each $\Lambda_{k,j}$ proportional
to $\lambda_{k+j}$ with integer coefficient.
\end{proof}
Matrix $\Lambda_g$ exhibits interesting properties of
sum of anti-diagonal elements implemented at derivations in \cite{eil16},
\begin{align}
\sum_{i,j, \;i+j=k}\Lambda_{g;i,j}=\lambda_k \frac{N_g}{4g+2}
\left[ \frac12 k(2g+2-k)+\frac14(2g+1)( (-1)^k-1 ) \right]
\end{align}

Lower genera examples of matrix $\Lambda$ were given in
\cite{ehkklp12}, \cite{eil16}, but method implemented there unable to
get expressions for $\Lambda$ at big genera.

\begin{example} At $g=6$ we get matrix
\begin{align}
\Lambda_6= \left(  \begin{array}{cccccc} 792\lambda_2&330\lambda_3&120\lambda_4&36\lambda_5&8\lambda_6&\lambda_7\\
                                       330\lambda_3&1080\lambda_4&492\lambda_5&184\lambda_6&51\lambda_7&8\lambda_8\\
                                       120\lambda_4&492\lambda_5&1200\lambda_6&542\lambda_7&184\lambda_8&36\lambda_9\\
                                       36\lambda_5&184\lambda_6&542\lambda_7&1200\lambda_8&492\lambda_9&120\lambda_{10}\\
                                       8\lambda_6&51\lambda_7&184\lambda_8&492\lambda_9&1080\lambda_{10}&330\lambda_{11}\\
                                       \lambda_7&8\lambda_8&36\lambda_9&120\lambda_{10}&330\lambda_{11}&792\lambda_{12}
\end{array}    \right)
\end{align}
\end{example}

Collecting all these together we get the following

\begin{proposition}  $\varkappa$-matrix defining multi-variate $\sigma$-function
admits the following modular form representation

\begin{align}
\varkappa=\frac{1}{8 N_g}\Lambda_g-\frac{1}{2N_g} {(2\omega)^{-1}}^T . \left[
\sum_{N_g\;\text{even}\;[\varepsilon]}   \frac{1}{\theta[\varepsilon]}
\left(\begin{array}{ccc}
\theta_{1,1}[\varepsilon]&\cdots&\theta_{1,g}[\varepsilon]\\
\vdots&\cdots&\vdots\\
\theta_{1,g}[\varepsilon]&\cdots&\theta_{g,g}[\varepsilon]\end{array}\right)\right].(2\omega)^{-1}
\label{formula2}
\end{align}
{\em where   $2\omega$ is matrix of $\mathfrak{a}$-periods of holomorhphic differentials and
$ \theta_{i,j}[\varepsilon]= \partial^2_{z_i,z_j}\theta[\varepsilon]({z})_{{z}=0} $.}
\end{proposition}

Note that modular form representation of period matrices $\eta, \eta'$ follows from the above formula,

\begin{equation}
\eta=2\varkappa\omega, \qquad \eta'=2\varkappa \omega'-\imath\pi{(2\omega)^T}^{-1}\label{etamodular}
\end{equation}

\begin{example} At $g=2$ for the curve $y^2=4x^5+\lambda_4x^4+\ldots+\lambda_0$
the following representation of $\varkappa$-matrix is valid
\begin{align}
\varkappa= \frac{1}{80}\left( \begin{array}{cc} 4\lambda_2&\lambda_3\\
\lambda_3&4\lambda_4 \end{array} \right)-\frac{1}{20}\sum_{10\;\text{even}\;[\varepsilon]}\frac{1}{\theta[\varepsilon]}
\left(\begin{array}{cc} \partial^2_{{U}_1^2}\theta[\varepsilon]&   \partial^2_{{U}_1,{U}_2}
\theta[\varepsilon]\\
 \partial^2_{{U}_1,{U}_2}
\theta[\varepsilon]& \partial^2_{{U}_2^2}\theta[\varepsilon]\end{array}\right)
\end{align}
with $\varkappa=\eta(2\omega)^{-1}$, $\mathcal{A}^{-1}=(2\omega)^{-1}=({U}_1,{U}_2 ) $ and
directional derivatives $\partial_{{U}_i}$, $i=1,2$.
\end{example}
Representation of $\varkappa$ matrix of genus 2 and 3 hyperelliptic curves in terms of directional derivatives
of non-singular odd constant was found in \cite{eee13}

\subsection{Co-homologies of Baker and Klein}
Calculations of $\varkappa$-matrix for the hyperelliptic curve (\ref{hyperelliptic}) were done
in co-homology basis introduced by H.Baker (\ref{bakerbasis}). When holomorphic differentials, $ \mathrm{d}{u} (x,y) $    are chosen meromorphic  differentials,          $ \mathrm{d}{r} (x,y)$ can be find from the symmetry condition $\mathbf{I}$. One can check that symmetry condition also fulfilled if meromorphic differentials will be changed as
\begin{equation}
 \mathrm{d}{r}(x,y) \rightarrow \mathrm{d}{r}(x,y)+M\mathrm{d}{u}(x,y),
\end{equation}
where $M$ is arbitrary constant symmetric matrix $M^T=M$. One can then choose
\begin{equation}
M=-\frac{1}{8 N_g}\Lambda_g
\end{equation}
Then $\varkappa$ will change to
\begin{align}
\varkappa=
-\frac1{2}\frac1{N_g}\sum_{N_g\;\text{even}\;\left[\varepsilon_{\mathcal{I}_0}\right]} \frac{1}{\theta[\varepsilon_{\mathcal{I}_0}]({0})} \left(\partial_{{U}_i}\partial_{{U_j}} \theta[\varepsilon_{\mathcal{I}_0}]({0})\right)_{i,j=1,\ldots,g}  .
\label{Kleinian}
\end{align}

Following \cite{eil16} introduce co-homology basis of Klein
\begin{align}
\mathrm{d}{u} (x,y), \quad  \mathrm{d}{r} (x,y)
-\frac{1}{8N_g}\Lambda_g\mathrm{d}{u} (x,y)\label{Kleinbasis}
\end{align}
with constant matrix, $\Lambda_g=\Lambda_g({\lambda})$ given by (\ref{Lambda1},\ref{Lambda2}). Therefore we proved

\begin{proposition} $\varkappa$-matrix is represented in the modular form (\ref{Kleinian}) in the co-homology
basis (\ref{Kleinbasis}).\end{proposition}

Formula (\ref{Kleinian})  first appears in F.Klein (\cite{klein886}, \cite{klein888}),
it was recently revisited in a more general context by Korotkin and Shramchenko (\cite{ksh12}) who
extended representation for $\varkappa$ to non-hyperelliptic curves. Correspondence of this representation
to the co-homology basis to the best knowledge of the authors was not earlier discussed.

Rewrite formula (\ref{Kleinian}) in equivalent form,
\begin{align}
\omega^T\eta=-\frac{1}{4N_g}
\sum_{N_g\;\text{even}\;[\varepsilon]}   \frac{1}{\theta[\varepsilon]}
\left(\begin{array}{ccc}
\theta_{1,1}[\varepsilon]&\cdots&\theta_{1,g}[\varepsilon]\\
\vdots&\cdots&\vdots\\
\theta_{1,g}[\varepsilon]&\cdots&\theta_{g,g}[\varepsilon]\end{array}\right)\label{omegaeta}
\end{align}
where $\omega,\eta$ are half-periods of holomorphic and meromorphic differentials in Kleinian basis.

\begin{example} For the Weierstrass cubic $y^2=4x^3-g_2x-g_3$ (\ref{omegaeta}) represents
Weierstrass relation (\ref{SecondKind}).\end{example}

\begin{example} At $g=2$ (\ref{omegaeta}) can be written in the form
\begin{align}
\omega^T\eta = -\frac{\imath \pi}{10} \left(\begin{array}{cc}
\partial_{\tau_{1,1}} & \partial_{\tau_{1,2}}\\
\partial_{\tau_{1,2}} & \partial_{\tau_{2,2}}
 \end{array}  \right)\; \mathrm{ln}\; \chi_{5}
\end{align}
where $\chi_5$ is relative invariant of weight 5,
\begin{align}
\chi_{5} = \prod_{ 10\;\text{even}\; [\varepsilon] } \theta[\varepsilon]
\end{align}
\end{example}
\begin{example} Worth to mention how equations of KdV flows looks in both bases.
 For example at $g=2$ and
curve $y^2=4x^5+\lambda_4x^4+\ldots+\lambda_0$ in Baker basis
we got \cite{bel997}
\begin{align}\begin{split}
\wp_{2222}&=6\wp_{2,2}^2+4\wp_{1,2}+\lambda_4\wp_{2,2}+\frac12\lambda_3   \\
\wp_{1222}&=6\wp_{2,2}\wp_{1,2}-2\wp_{1,1}+\lambda_4\wp_{1,2} \end{split}
\end{align}
In Kleinian basis the same equations change only in linear in $\wp_{i,j}$-terms
\begin{align}\begin{split}
\wp_{2222}&=6\wp_{2,2}^2+4\wp_{1,2}-47\lambda_4\wp_{2,2}+92\lambda_4^2-\frac72\lambda_3\\
\wp_{1222}&=6\wp_{2,2}\wp_{1,2}-2\wp_{1,1}-23\lambda_4\wp_{1,2}-6\lambda_3\wp_{2,2}+23\lambda_3\lambda_4+8\lambda_2 \end{split}
\end{align}
\end{example}

\subsection{Rosenhain modular form representaion of first kind periods}

Rosenhain \cite{ros851}  was the first who introduced $\theta$-functions with characteristics at $g=2$.
There are 10 even and 6 odd characteristics in that case. Let us denote each from these characteristics as
\[  \varepsilon_j = \left[\begin{array}{c} {\varepsilon_j'}^T\\ {\varepsilon''_j}^T \end{array}   \right],
\quad j=1,\ldots 10  \]
where $\varepsilon'_j$ and $\varepsilon''_j$ are column 2-vectors with entries equal to 0 or 1.

Rosenhain fixed the hyperelliptic genus two curve in the form
\[ y^2=x(x-1)(x-a_1)(x-a_2)(x-a_3)   \]
and presented  without proof expression
\begin{align}
\mathcal{A}^{-1} =\frac{1}{2\pi^2 Q^2 } \left( \begin{array}{rr}   -P\theta_2[\delta_2]&Q\theta_2[\delta_1]\\
 P\theta_1[\delta_2]&-Q\theta_1[\delta_1]  \end{array}    \right) \label{RosenhainFormula}
\end{align}
with
\begin{align*}
P&=\theta[\alpha_1]\theta[\alpha_2]\theta[\alpha_3],\quad
Q=\theta[\beta_1]\theta[\beta_2]\theta[\beta_3]
\end{align*}
and 6 even characteristics $[\alpha_{1,2,3}], [\beta_{1,2,3}]$ and two odd $[\delta_{1,2}]$ which looks chaotic.
One of first proofs can be found in H.Weber \cite{web859}; these formulae are implemented in Bolza dissertation
\cite{bol885} and \cite{bol886}. Our derivation of these formulae are based on the {\em Second Thomae relation}
\cite{tho870}, see \cite{er08} and \cite{eil18}. To proceed we give the following definitions.

\begin{definition} A triplet of characteristics $[\varepsilon_1]$, $[\varepsilon_2,]$,
$[\varepsilon_3]$ is called azygetic if

\[ \mathrm{exp}\; \imath \pi  \left\{
 \displaystyle{\sum_{j=1}^3{\varepsilon_j'}^T \varepsilon''_j + \sum_{i=1}^3{\varepsilon_i'}^T
\sum_{i=1}^3\varepsilon''_i } \right\} =- 1  \]
\end{definition}

\begin{definition} A sequence of $2g + 2$ characteristics
$[\varepsilon_1],\ldots,[\varepsilon_{2g+2} ] $ is called
a {\em  special fundamental system} if the first $g$ characteristics are odd, the
remaining are even and any triple of characteristics in it is azygetic.
\end{definition}

\begin{theorem}(Conjectural Riemann-Jacobi derivative formula)
Let $g$ odd $[\varepsilon_1], \ldots,  [\varepsilon_g]$ and $g+2$ even $ [\varepsilon_{g+1}], \ldots, [\varepsilon_{2g+2}]$ characteristics create a special fundamental system. Then the following equality is valid
\begin{equation}\mathrm{Det}\left.  \frac{ \partial( \theta[\varepsilon_1]({v}),\ldots,
			\theta[\varepsilon_g]({v}))  }
		{\partial (v_1,\ldots, v_g)}\right|_{{v}=0}= \pm    \prod_{k=1\ldots g+2}\theta[\varepsilon_{g+k}]({0})
\label{RiemannJacobi}  \end{equation}
\end{theorem}
\begin{proof} (\ref{RiemannJacobi})  proved up to $g=5$ \cite{fro885}, \cite{igu980}, \cite{fay979}
\end{proof}

\begin{example} Jacobi derivative formula for elliptic curve
\begin{align*}
\vartheta'_1(0)=\pi \vartheta_2(0)\vartheta_3(0)\vartheta_4(0)
\end{align*}
\end{example}

\begin{example} Rosenhain derivative formula for genus two curve is given without proof in the memoir
\cite{ros851}, namely, let $[\delta_1]$ and $[\delta_2]$ are any two odd characteristics from all 6 odd, then
\begin{align}
 \theta_1[\delta_1] \theta_2[\delta_2]-\theta_2[\delta_1] \theta_1[\delta_2]
= \pi^2\theta[\gamma_1]\theta[\gamma_2]\theta[\gamma_3]\theta[\gamma_4]
\label{RosenhainDerivative}\end{align}
where 4 even characteristics $[\gamma_1], \ldots, [\gamma_4]$ are given as
$[\gamma_i]=[\delta_1]+[\delta_2]+[\delta_i], \; 3\leq i \leq 6$. There are 15 Rosenhain derivative formulae.
\end{example}

The following geometric interpretation of the special fundamental system  can be given in the case of hyperelliptic curve.
Consider genus two curve,
\[\mathcal{C}:\quad   y^2= (x-e_1) \cdots ( x-e_{6})   \]
Denote associated homology basis as
$( \mathfrak{a}_1, \mathfrak{a}_2;\mathfrak{b}_1,\mathfrak{b}_2 )$.
Denote characteristics of
Abelian images of branch points with base point $e_6$ as $ \mathfrak{A}_k$, $k=1,\ldots,6$.
These are half-periods given by their characteristics, $[\mathfrak{A}_k]$ with
\begin{align}
\mathfrak{A}_k=
\int_{(e_6,0)}^{(e_k,0)} {u}= \frac12 \tau {\varepsilon'}_k
+ \frac12 {\varepsilon''}_k, \quad k=1,\ldots,6
\end{align}

For the homology basis drawn on the Figure we have


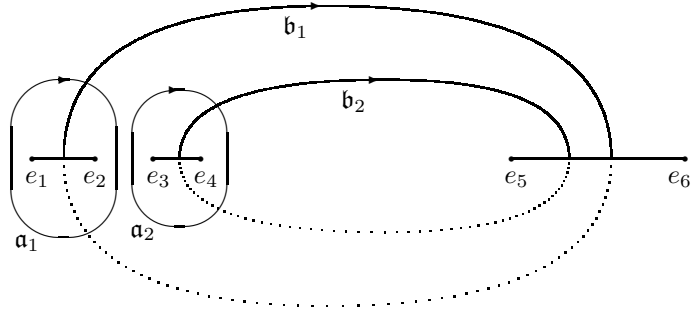
\begin{figure}
\begin{center}
\unitlength 0.7mm \linethickness{0.6pt}
\begin{picture}(150.00,80.00)
\put(9.,33.){\line(1,0){12.}} \put(9.,33.){\circle*{1}}
\put(21.,33.){\circle*{1}} \put(10.,29.){\makebox(0,0)[cc]{$e_1$}}
\put(21.,29.){\makebox(0,0)[cc]{$e_2$}}
\put(15.,33.){\oval(20,30.)}
\put(8.,17.){\makebox(0,0)[cc]{$\mathfrak{ a}_1$}}
\put(15.,48.){\vector(1,0){1.0}}
\put(32.,33.){\line(1,0){9.}} \put(32.,33.){\circle*{1}}
\put(41.,33.){\circle*{1}} \put(33.,29.){\makebox(0,0)[cc]{$e_3$}}
\put(42.,29.){\makebox(0,0)[cc]{$e_4$}}
\put(37.,33.){\oval(18.,26.)}
\put(30.,19.){\makebox(0,0)[cc]{$\mathfrak{a}_2$}}
\put(36.,46.){\vector(1,0){1.0}}
\put(100.,33.00) {\line(1,0){33.}} \put(100.,33.){\circle*{1}}
\put(133.,33.){\circle*{1}}
\put(101.,29.){\makebox(0,0)[cc]{$e_{5}$}}
\put(132.,29.){\makebox(0,0)[cc]{$e_{6}$}}
\put(59.,58.){\makebox(0,0)[cc]{$\mathfrak{b}_1$}}
\put(63.,62.){\vector(1,0){1.0}}
\bezier{484}(15.,33.00)(15.,62.)(65.,62.)
\bezier{816}(65.00,62.)(119.00,62.00)(119.00,33.00)
\bezier{35}(15.,33.00)(15.,5.)(65.,5.)
\bezier{35}(65.00,5.)(119.00,5.00)(119.00,33.00)
\put(70.,44.){\makebox(0,0)[cc]{$\mathfrak{b}_2$}}
\put(74.00,48.){\vector(1,0){1.0}}
\bezier{384}(37.,33.00)(37.,48.)(76.00,48.)
\bezier{516}(76.00,48.)(111.00,48.00)(111.00,33.00)
\bezier{30}(37.,33.00)(37.,19.)(76.00,19.)
\bezier{30}(76.00,19.)(111.00,19.00)(111.00,33.00)
\end{picture}
\end{center}
\caption{Homology basis on the Riemann surface of the curve
$\mathcal{C}$ with real branching points $e_1 < e_2 <\ldots <
e_{6}$ (upper sheet).  The cuts are drawn from $e_{2i-1}$
to $e_{2i}$, $i=1,2,3$.  The $ \mathfrak{b}$--cycles are completed on the
lower sheet (dotted lines).}
\end{figure}

\begin{align*} [{\mathfrak A}_1]= \frac{1}{2}
\left[\begin{array}{cc}1&0\\
                       0&0\end{array}\right],\quad
[{\mathfrak A}_2]= \frac{1}{2}
\left[\begin{array}{cc}1&0\\
                       1&0\end{array}\right],\quad
[{\mathfrak A}_3]= \frac{1}{2}
\left[\begin{array}{cc}0&1\\
                       1&0\end{array}\right],
\end{align*}
\begin{align*}
[{\mathfrak A}_4]= \frac{1}{2}
\left[\begin{array}{cc}0&1\\
                       1&1\end{array}\right],\quad
[{\mathfrak A}_5]= \frac{1}{2}
\left[\begin{array}{cc}0&0\\
                       1&1\end{array}\right],\quad
[{\mathfrak A}_6]= \frac{1}{2}
\left[\begin{array}{cc}0&0\\
                       0&0\end{array}\right]
\end{align*}
\[ {\mathfrak{A}}_k= \int_{\infty}^{e_k}{u}=   \frac12 \tau.
{\varepsilon}_k
+ \frac12 {\varepsilon'}_k , \quad [ {\mathfrak{A}}_k ] =
 \left[ \begin{array}{c}
 {\varepsilon}_k^T\\   {\varepsilon'}_k^T \end{array} \right]     \]
 One can see that the set of characteristics $[\varepsilon_k] = [ \mathfrak{A}_k]$ of Abelian images of branch point
contains two odd $[\varepsilon_2]$ and $[\varepsilon_4]$  and remaining four characteristics are even.
One can check that the whole set of these $6=2g+2$  characteristics is azygetic and therefore the set constitutes
special fundamental system. Hence one can write Rosenhain derivative formula
\begin{equation}
\theta_1[\varepsilon_2]\theta_2[\varepsilon_4]-\theta_2[\varepsilon_2]\theta_1[\varepsilon_4]=\pm \pi^2
\theta[\varepsilon_1]\theta[\varepsilon_3]\theta[\varepsilon_5]\theta[\varepsilon_6]
\end{equation}
In this way we gave geometric interpretation of Rosenhain derivative formula and associated
set of characteristics in the case when one from even characteristic is zero. The same structure
is observed for higher genera hyperelliptic curves even for $g>5$.

\begin{proposition}The characteristics entering to the Rosenhain formula are described as follows.
Take any of 15 Rosenhain derivative formula, say,
\[  \theta_1[p] \theta_2[q]-\theta_2[p] \theta_1[q]
= \pi^2\theta[\gamma_1]\theta[\gamma_2]\theta[\gamma_3]\theta[\gamma_4]   \]
Then 10 even characteristics can be grouped as
\[ \underbrace{ [\gamma_1],\ldots,[\gamma_4]}_4, \; \underbrace{ [\alpha_1],[\alpha_2],[\alpha_3] }_{[\alpha_1]+ [\alpha_2]+ [\alpha_3]=[p]},\; \;
\underbrace{ [\beta_1],[\beta_2],[\beta_3] }_{ [\beta_1]+ [\beta_2]+ [\beta_3]=[q]},   \]
Then matrix of $\mathfrak{a}$-periods
\begin{align*}
\mathcal{A} &= \frac{2Q}{PR} \left( \begin{array}{rr}
Q\theta_1[q]&Q \theta_2[q]\\
P\theta_1[p]& P\theta_2[p]
     \end{array}  \right)\end{align*}
with
\begin{align}
P=\prod_{j=1}^3\theta[\alpha_j], \;\;Q=\prod_{j=1}^3\theta[\beta_j],\; \;R=\prod_{j=1}^4\theta[\gamma_j]\label{PQR}
\end{align}
\end{proposition}
Note, that the 15 curves are given as
\begin{align} \mathcal{C}_{p,q}:   \quad y^2=x(x-1)(x-a_1)(x-a_2)(x-a_3) \label{Bolza}   \end{align}
where branch points are computed by  {\em Bolza formulae} \cite{bol886},
\begin{align} e[\delta_j]=  - \frac{\partial_{{U}_1} \theta[\delta_j] }{\partial_{\boldsymbol{U}_2} \theta[\delta_j]}  , \quad
j=1,\ldots,6
\label{bolza1}
  \end{align}
where $\partial_{{U}_i}$ directional derivative along vector ${U}_i$, $i=1,2$, $\mathcal{A}^{-1}=({U}_1,{U}_2)$ and
\[e[p]=0,\; e[q]=\infty,\      \]
All 15 curves  $\mathcal{C}_{p,q}$ are M\"obius equivalent.

\subsection{ \underline{Generalization of the Rosenhain formula to higher
genera hyperelliptic curve} }

Generalization of the Rosenhain formula to higher
genera hyperelliptic curve was found in \cite{er08} and developed further in \cite{eil18}.

\begin{align}\begin{split} y^2&=\phi(x) \psi(x)\\ \phi(x)&= \prod_{k=1}^g(x-e_{2k}), \;
\psi(x)= \prod_{k=1}^{g+1}(x-e_{2k-1}) \label{curve-g}  \end{split}   \end{align}

Denote $ R=\prod_{k=1}^{g+2}\theta[\gamma_k]$
monomial in left had side of  Riemann-Jacobi formula (\ref{RiemannJacobi})
\begin{proposition}Let genus $g$ hyperelliptic curve is given in (\ref{curve-g}). Then winding vectors
$(U_1,\ldots,U_g)=\mathcal{A}^{-1}$ are given by the formula
\begin{align}
{U}_m=\frac{\epsilon}{ 2\pi^g R} \mathrm{Cofactor}\left(  \left.  \frac{ \partial( \theta[\varepsilon_1]({v}),\ldots,
			\theta[\varepsilon_g]({v}))  }
		{\partial (v_1,\ldots, v_g)}\right|_{{v}=0}
 \right)
\left( \begin{array}{c}  s_{m-1}^{2} \sqrt[4]{\chi_1}\\ \vdots
\\s_{m-1}^{2g} \sqrt[4]{ \chi_g}\end{array}  \right)
\end{align}
Here $s_k^i$ -order $k$ symmetric function  of elements $\{e_{2}, \ldots e_{2g}  \}  /  \{ e_{2i} \}$ and
\[  \chi_{i} = \frac{\psi(e_{2i})}{\phi'(e_{2i})} , \quad i=1,\ldots, g \]
$s_k^i, \chi_i$ are expressible in $\theta$-constants via Thomae formulae \cite{tho870}.
\end{proposition}

\subsection{\underline{Applications of the above results}}

Typical answer of multi-gap integration includes $\theta$-function
$  \theta({U}x +{V}t+{W};\tau )$
where winding vectors ${U}, {V}$ expressed
in terms of complete holomorphic integrals and constant ${W}$ defined by initial data.
Rosenhain formulae and their generalization, express  $U,V$ in terms of
$\theta$-constants and parameters the equation defining $\mathcal{C}$. In this way
the problem of {\em effectivization of finite gap solutions} \cite{dub81}
can be solved in that way at least for hyperelliptic curves.

Other application of the Rosenhain formula (\ref{RosenhainFormula}) presented in \cite{be994}
where two-gap Lam\'e and Treibich-Verdier potentials were obtained by the reduction
to elliptic functions of general Its-Matveev  representation \cite{im975}  of finite-gap potential
to the Schr\"odinger equation in terms of multi-variable $\theta$-functions.

Another application relevant to a computer algebra problem. In the case when, say in Maple, periods of
holomorphic differentials are computed them periods of second kind differentials can be obtained by
Rosenhain formula (\ref{RosenhainFormula}) and its generalization.


\section{Sigma-functions and the problem of differentiatiion of Abelian functions.}


\subsection{\underline{Problems and methods}}

Consider the curve
\begin{equation} \label{F-1}
V_\lambda = \left\{(x,y)\in\mathbb{C}^2\, :  y^2 = \mathcal{C}(x;\lambda) = x^{2g+1}+\sum_{k=2}^{2g+1}
\lambda_{2 k} x^{2g - k + 1} \right\}
\end{equation}
where $g\geqslant 1$ and $\lambda=(\lambda_4,\ldots,\lambda_{4g+2})\in \mathbb{C}^{2g}$ are the parameters.
Set $\mathcal{D} = \{ \lambda\in \mathbb{C}^{2g}\,: \mathcal{C}(x;\lambda)\; \text{has multiple roots}\}$
and $\mathcal{B} = \mathbb{C}^{2g}\setminus\mathcal{D}$.
For any $\lambda\in \mathcal{B}$ we obtain the affine part of a smooth projective hyperelliptic curve
$\overline{V}_\lambda$ of genus~$g$ and the Jacobian variety $Jac(\overline{V}_\lambda) = \mathbb{C}^{g}/\Gamma_g$,
where $\Gamma_g \subset \mathbb{C}^{g}$  is a lattice of rank $2g$ generated by the periods of the holomorphic
differential on cycles of the curve $V_\lambda$.

In the general case, an \emph{Abelian function} is
a meromorphic function on a complex Abelian torus $T^g=\mathbb{C}^g\!/\Gamma$, where
$\Gamma\subset\mathbb{C}^g$ is a lattice of rank $2g$. In other words, a meromorphic function $f$ on
$\mathbb{C}^g$ is Abelian iff $f(u)=f(u+\omega)$ for all
$u=(u_1,\ldots,u_g)\in\mathbb{C}^{g}$ and $\omega\in\Gamma$. Abelian functions on $T^g$ form a field  $\mathcal{F} = \mathcal{F}_g$ such that:

(1) let $f\in\mathcal{F}$, then $\partial_{u_i} f\in \mathcal{F}$, $i=1,\dots,g$;

(2) let $f_1,\dots,f_{g+1}$  be any  nonconstant functions from $\mathcal{F}$, then there exists a polynomial
$P$ such that $P(f_1,\dots, f_{g+1})(u)=0$ for all $u\in T^g$;

(3) let $f\in\mathcal{F}$ be a nonconstant function, then any  $h\in\mathcal{F}$ can be expressed rationally
in terms of $(f,\partial_{u_1}f,\dots,\partial_{u_g}f)$;

(4) there exists an entire function $\vartheta\colon\mathbb{C}^g\to\mathbb{C}$ such that
$\partial_{u_i,u_j}\log\vartheta \in \mathcal{F}$, $i,j=1,\dots,g$.

For example, any elliptic function $f\in \mathcal{F}_1$ is a rational function in the Weierstrass functions
$\wp(u; g_2, g_3)$ and $\partial_u\wp(u; g_2, g_3)$, where $g_2$ and $g_3$ are parameters of elliptic curve
\[
V = \{ (x,y) \in \mathbb{C}^2\;|\; y^2 = 4 x^3 - g_2 x - g_3 \}.
\]
It is easy to see that the function $\frac{\partial}{\partial g_2}\wp(u; g_2, g_3)$ will no longer be elliptic.
This is due to the fact that the period lattice $\Gamma$ is a function of the parameters $g_2$ and $g_3$.
In \cite{fs882} Frobenius and Stickelberger described all the differential operators $L$ in the variables
$u,\;g_2$ and $g_3$, such that $Lf\in \mathcal{F}_1$ for any function $f\in \mathcal{F}_1$ (see below Section 7.3).

In \cite{bl07, bl08} the classical problem of differentiation of Abelian functions over parameters
for families of $(n,s)$-curves was solved.
In the case of hyperelliptic curves this problem was solved more explicitly.

All genus $2$ curves are hyperelliptic. We denote by $\pi\colon \mathcal{U}_g \to \mathcal{B}_g$
the universal  bundle of  Jacobian varieties $Jac(\overline{V}_\lambda)$ of hyperelliptic curves.
Let us consider the mapping $\varphi\colon \mathcal{B}_g\times \mathbb{C}^g \to \mathcal{U}_g$, which defines
the projection $\lambda\times \mathbb{C}^g \to \mathbb{C}^g/\Gamma_g(\lambda)$ for any $\lambda\in \mathcal{B}_g$.
Let us fix the coordinates $(\lambda;u)$ in $\mathcal{B}_g\times \mathbb{C}^g\subset \mathbb{C}^{2g}\times \mathbb{C}^g$
where $u = (u_1,\ldots,u_{2g-1})$. Thus, using the mapping $\varphi$, we fixed in $\mathcal{U}_g$ the structure
of the space of the bundle whose fibers $J_\lambda$ are principally polarized Abelian varieties.

We denote by $F = F_g$ the field of functions on $\mathcal{U}_g$ such that for any $f\in F$ the function $\varphi^*(f)$
is meromorphic, and its restriction to the fiber $J_\lambda$ is an Abelian function for any point $\lambda \in \mathcal{B}_g$.

Below, we will identify the field $F$ with its image in the field of meromorphic functions on $\mathcal{B}\times \mathbb{C}^g$.

The following {\bf Problem I}:

{\it Describe the Lie algebra of differentiations of the field of meromorphic functions on $\mathcal{B}_g\times \mathbb{C}^g$,
generated by the operators $L$, such that $Lf\in F$ for any function $f\in F$}

was solved in \cite{bl07, bl08}.

From the differential geometric point of view, Problem I is closely related to {\bf Problem II}:

{\it Describe the connection of the bundle $\pi\colon \mathcal{U}_g \to \mathcal{B}_g$.}

The solution of Problem II leads to an important class of solutions of well-known equations of mathematical physics.
In the case $g=1$, the solution is called the Frobenius-Stikelberger connection (see \cite{dubr96}) and
leads to solutions of Chazy equation.

The space $\mathcal{U}_g$ is a rational variety, more precisely, there is a birational isomorphism $\varphi\colon \mathbb{C}^{3g} \to \mathcal{U}_g$.
This fact was discovered by B.~A.~Dubrovin and S.~P.~Novikov  in \cite{dn974}. In \cite{dn974},
a fiber of the universal bundle is considered as a level surface of the integrals of motion of $g$th
stationary flow of KdV system, that is, it is defined in $\mathbb{C}^{3g}$ by a system of $2g$ algebraic equations.
The degree of the system grows with the growth of genus. In \cite{bel97}, \cite{bel997}  the coordinates in $\mathbb{C}^{3g}$
were introduced such that a fiber is defined by $2g$ equations of degree not greater than~$3$.

The Dubrovin-Novikov  coordinates and the coordinates from
\cite{bel97}, \cite{bel997} are the same for the universal space of genus $1$ curves. But already in the case of genus 2,
these coordinates differ (see \cite{bl08}).

The integrals of motion of KdV systems are exactly the coefficients $\lambda_{2g+4},\ldots,\lambda_{4g+2}$ of hyperelliptic
curve $V_\lambda$, in which the coefficients $\lambda_{4},\ldots,\lambda_{2g+2}$ are free parameters (see \eqref{F-1}).
Choosing a point $z\in \mathbb{C}^{3g}$ such that the point $\varphi(z)\in \mathcal{B}_g$ is defined,
one can calculate the values of the coefficients $(\lambda_4,\ldots,\lambda_{4g+2}) = \pi(z)\in \mathcal{B}_g$ substituting
this point in these integrals. Thus, the solution of the Problem I of differentiation of hyperelliptic functions led to
the solution of another well-known {\bf Problem III}:

{\it Describe the dependence of the solutions of $g$-th stationary flow of KdV system on the variation of the coefficients
$\lambda_4,\ldots,\lambda_{4g+2}$ of hyperelliptic curve, that is, from variation of values of the integrals of motion and parameters.}

In \cite{go89} it was obtained results on Problem III, which use the fact that for a hierarchy KdV the action of polynomial
vector fields on the spectral plane is given by the shift of the branch points of the hyperelliptic curve along this fields
(see \cite{ss14}). The deformations of the potential corresponding to this action are exactly the action
of the nonisospectral symmetries of the hierarchy KdV.

Let us describe a different approach to Problem III, developed in our works.
In \cite{bl02} it was introduced the concept of a polynomial Lie algebra over a ring of polynomials $A$.
For brevity, we shall call them Lie A-algebras.
In \cite{bl07, bl08} it was considered the ring of polynomials $\mathcal{P}$ in the field $F$.
This ring is generated by all logarithmic derivatives of order $k\geqslant 2$ from the hyperelliptic sigma function $\sigma(u;\lambda)$.
It was constructed the Lie $\mathcal{P}$-algebra $\mathcal{L}=\mathcal{L}_g$ with generators $L_{2k-1},\; k=1,\ldots,g$
and $L_{2l},\; l=0,\ldots,2g-1$.
The fields $L_{2k-1}$ define isospectral symmetries, and the fields $L_{2l}$ define nonisospectral symmetries of the hierarchy KdV.
The Lie algebra $\mathcal{L}$ is isomorphic to the Lie algebra of differentiation of the ring $\mathcal{P}$  and, consequently,
allows to solve the Problem I (see property (3) of Abelian functions). The generators $L_{2k-1},\; k=1,\ldots,g$, coincide
with the operators $\partial_{u_{2k-1}}$, and, consequently, commute. Thus, in the Lie $\mathcal{P}$-algebra $\mathcal{L}$
it is defined the Lie $\mathcal{P}$-subalgebra $\mathcal{L}^*$ generated by the operators $L_{2k-1},\; k=1,\ldots,g$.
The generators $L_{2l},\; l=0,\ldots,2g-1,$ are such that the Lie $\mathcal{P}$-algebra $\mathcal{L}^*$ is an ideal
in the Lie $\mathcal{P}$-algebra $\mathcal{L}$.

The construction of $L_{2k},\; k=0,\ldots,2g-1$, is based on the following fundamental fact (see \cite{bl04}):

The entire function $\psi(u;\lambda)$, satisfying the system of heat equations in a nonholonomic frame
\[
\ell_{2i}\psi = H_{2i}\psi,\; i=0,\ldots,2g-1,
\]
under certain initial conditions (see \cite{bl04}) coincides with the hyperelliptic sigma-function $\sigma(u;\lambda)$.
Here $\ell_{2i}$ are polinomial linear first-order differential operators in the variables $\lambda=(\lambda_4,\ldots,\lambda_{4g+2})$
and $H_{2i}$ are linear second-order differential operators in the variables $u=(u_1,\ldots,u_{2g-1})$.
The methods for constructing these operators are described in \cite{bl04}.

The following fact was used  essentially in constructing the operators $\ell_{2i}$:

The Lie $\mathbb{C}[\lambda]$-algebra $\mathcal{L}_\lambda$ with generators $\ell_{2i},\; i=0,\ldots,2g-1$
is isomorphic to an infinite-dimensional Lie algebra $Vect_{\mathcal{B}}$ of vector fields
on $\mathbb{C}^{2g}$, that are tangent to the discriminant variety $\Delta$.
We recall that the Lie algebra $Vect_{\mathcal{B}}$ is essentially used in singularity theory and its applications
(see \cite{ar90}).

In the Lie $\mathcal{P}$-algebra $\mathcal{L}$ we can choose the generators $L_{2k},\; k=0,\ldots,2g-1$ such that
for any polynomial $P(\lambda)\in \mathbb{C}[\lambda]$ and any $k$ the formula $L_{2k}\pi^*P(\lambda) = \pi^*(\ell_{2k}P(\lambda))$
holds, where $\pi^*$ is the ring homomorphism induced by the projection $\pi \colon \mathcal{U}_g \to \mathcal{B}_g$.

Section 3 describes the development of an approach to solving the Problem I. This approach uses:
\begin{enumerate}
  \item[(a)] the graded set of multiplicative generators of the polynomial ring $\mathcal{P}=\mathcal{P}_g$;
  \item[(b)] the description of all algebraic relations between these generators;
  \item[(c)] the description of the birational isomorphism $J \colon \mathcal{U}_g \to \mathbb{C}^{3g}$ in terms of graded polynomial rings;
  \item[(d)] the description of the polynomial projection $\pi \colon \mathbb{C}^{3g} \to \mathbb{C}^{2g}$, where $\mathbb{C}^{2g}$
  is a space in  coordinates $\lambda = (\lambda_4,\ldots,\lambda_{4g+2})$, such that for any $\lambda \in \mathcal{B}$ the space
  $J^{-1}\pi^{-1}(\lambda)$  is the Jacobian variety $Jac(V_\lambda)$;
  \item[(e)] the construction of linear differential operators of first order
  \[
  \widehat{H}_{2k} = \sum_{i=1}^q h_{2(k-i)+1}(\lambda;u)\partial_{u_{2i-1}}, \; q=\min(k,g),
  \]
   such that $L_{2k} = \ell_{2k}-\widehat{H}_{2k},\; k=0,\ldots,2g-1$.
\end{enumerate}
Here:

- $h_{2(k-i)+1}(\lambda;u)$ are meromorphic functions on $\mathcal{B}_g\times \mathbb{C}^{g}$;

- $h_{2(k-i)+1}(\lambda;u)$ are homogeneous functions of degree $2(k-i)+1$ in $\lambda = (\lambda_4,\ldots,\lambda_{4g+2})$,\;
$\deg\lambda_{2k}=2k$, and $u = (u_1,\ldots,u_{2g-1}),\; \deg u_{2k-1}=1-2k$;

- $\partial_{u_{2l-1}}h_{2(k-i)+1}$ are homogeneous polynomials of degree $2(k+l-i)$ in the ring $\mathcal{P}_g$.

This approach was proposed in \cite{buch16} and found the application in \cite{bun17}.
A detailed construction of the Lie algebra $\mathcal{L}_2$ is given in \cite{buch16}, and the Lie algebra $\mathcal{L}_3$ in \cite{bun17}.

General methods and results (see Section 3.2) will be demonstrated in cases $g=1$ (see Section 3.3) and $g=2$ (see Section 3.4).

\subsection{\underline{Hyperelliptic functions of genus $g\geqslant 1$}}

For brevity, Abelian functions on the Jacobian varieties of \eqref{F-1}  will be called {\it hyperelliptic functions} of genus~$g$.
In the theory and applications of these functions, that are based on the sigma-function
$\sigma(u;\lambda)$ (see \cite{bak998, bel97, bel997, bel12}), the grading plays an important role.
Below, the variables $u=(u_1,u_3,\ldots,u_{2g-1})$, parameters $\lambda=(\lambda_4,\ldots,\lambda_{4g+2})$ and functions
 are indexed in a way that clearly indicates their grading. Note that our new notations for
the variables differ from the ones in \cite{bel97, bel997, bel12} as follows
\[
u_i \longleftrightarrow u_{2(g-i)+1},\, i=1,\ldots,g.
\]

Let
\[
\omega = \left( (2k_1-1)\cdot j_1,\ldots,(2k_s-1)\cdot j_s \right)
\]
where $1\leqslant s\leqslant g$,\; $j_q> 0,\; q=1,\ldots,s$ and $j_1+\ldots+j_s\geqslant 2$.
We draw attention to the fact that the symbol ``$\cdot$'' in the two-component expression $(2k_q-1)\cdot j_q$ is not
a multiplication symbol.
Set
\begin{equation} \label{f-2}
\wp_\omega(u;\lambda)=
-\partial^{j_1}_{u_{2k_1-1}} \cdots \partial^{j_s}_{u_{2k_s-1}}\,\ln\sigma(u;\lambda).
\end{equation}
Thus
\[
\deg \wp_\omega = (2k_1-1)j_1+\cdots+(2k_s-1)j_s.
\]

Note that our $\omega$ differ from the ones in \cite{buch16, bun17}.

Say that a multi-index $\omega$ is given in normal form if $1 \leqslant k_1 < \ldots < k_s$. According to formula \eqref{f-2},
we can always bring the multi-index $\omega$ to a normal form using the identifications:
\begin{align*}
  \left( (2k_p-1)\cdot j_p,(2k_q-1)\cdot j_q  \right) &= \left( (2k_q-1)\cdot j_q,(2k_p-1)\cdot j_p  \right), \\
  \left( (2k_p-1)\cdot j_p,(2k_q-1)\cdot j_q  \right) &= (2k_p-1) \cdot (j_p+j_q),\; \text { if }\, k_p=k_q.
\end{align*}

In \cite{bel12} (see also \cite{bel97, bel997}) it was proved that for $1\leqslant i \leqslant k \leqslant g$ all algebraic relations
between hyperelliptic functions of genus $g$ follow from the relations, which in our graded notations have the form
\begin{equation} \label{f-3}
\wp_{1\cdot 3,(2i-1)\cdot 1} = 6\left(\wp_{1\cdot 2}\wp_{1\cdot 1,(2i-1)\cdot 1} + \wp_{1\cdot 1,(2i+1)\cdot 1}\right) -
2\left(\wp_{3\cdot 1,(2i-1)\cdot 1} - \lambda_{2i+2} \delta_{i,1}\right).
\end{equation}
Here and below, $\delta_{i,k}$ is the Kronecker symbol, $\deg \delta_{i,k}=0$.
\begin{multline}\label{f-4}
\wp_{1\cdot 2,(2i-1)\cdot 1}\wp_{1\cdot 2,(2k-1)\cdot 1} = 4\left(\wp_{1\cdot 2}\wp_{1\cdot 1,(2i-1)\cdot 1}\wp_{1\cdot 1,(2k-1)\cdot 1} + \wp_{1\cdot 1,(2k-1)\cdot 1}\wp_{1\cdot 1,(2i+1)\cdot 1}\right. +\\
\left.+\wp_{1\cdot 1,(2i-1)\cdot 1}\wp_{1\cdot 1,(2k+1)\cdot 1} + \wp_{(2k+1)\cdot 1,(2i+1)\cdot 1} \right) - 2\left(\wp_{1\cdot 1,(2i-1)\cdot 1}\wp_{3\cdot 1,(2k-1)\cdot 1} \right. + \\
\left. + \wp_{1\cdot 1,(2k-1)\cdot 1}\wp_{3\cdot 1,(2i-1)\cdot 1} +\wp_{(2k-1)\cdot 1,(2i+3)\cdot 1} +\wp_{(2i-1)\cdot 1,(2k+3)\cdot 1} \right) + \\
+ 2 \left(\lambda_{2i+2}\wp_{1\cdot 1,(2k-1)\cdot 1}\delta_{i,1} + \lambda_{2k+2}\wp_{1\cdot 1,(2i-1)\cdot 1}\delta_{k,1} \right)  +
2\lambda_{2(i+j+1)} (2\delta_{i,k} + \delta_{k,i-1} + \delta_{i,k-1}).
\end{multline}

\begin{corollary}\label{ex-1}
For all $g\geqslant 1$, we have the formulas:
\begin{enumerate}
  \item [1.] Setting $i=1$ in \eqref{f-3}, we obtain
\begin{equation}\label{ex-1-1}
\wp_{1\cdot 4} = 6\wp_{1\cdot 2}^2 + 4\wp_{1\cdot 1,3\cdot 1} + 2\lambda_4.
\end{equation}
  \item [2.] Setting $i=2$ in \eqref{f-3}, we obtain
\begin{equation}\label{ex-1-2}
\wp_{1\cdot 3,3\cdot 1} = 6(\wp_{1\cdot 2} \wp_{1\cdot 1,3\cdot 1} + \wp_{1\cdot 1,5\cdot 1}) - 2\wp_{3\cdot 2}.
\end{equation}
  \item [3.] Setting $i=k=1$ in \eqref{f-4}, we obtain
\begin{equation}\label{ex-1-3}
\wp_{1\cdot 3}^2 = 4\left[ \wp_{1\cdot 2}^3 + (\wp_{1\cdot 1,3\cdot 1} + \lambda_4)\wp_{1\cdot 2} + (\wp_{3\cdot 2} -
\wp_{1\cdot 1,5\cdot 1} + \lambda_6) \right].
\end{equation}
\end{enumerate}
\end{corollary}

\begin{theorem}\label{T-7.1}
{\rm 1.} For any $\omega = \left( (2k_1-1)\cdot j_1,\ldots,(2k_s-1)\cdot j_s \right)$ the hyperelliptic function
$\wp_\omega(u;\lambda)$ is a polynomial from $3g$ functions $\wp_{1\cdot j,(2k-1)\cdot 1},\;1\leqslant j \leqslant 3,\;
1\leqslant k\leqslant g$.

Note that if $k=1$, we have $\wp_{1\cdot j,1\cdot 1} = \wp_{1\cdot (j+1)}$.

{\rm 2.} Set $W_{\wp} = \{ \wp_{1\cdot j,(2k-1)\cdot 1},\;1\leqslant j \leqslant 3,\; 1\leqslant k\leqslant g \}$.
The projection of the universal bundle  $\pi_g\colon \mathcal{U}_g \to \mathcal{B}_g \subset \mathbb{C}^{2g}$
is given by the polynomials $\lambda_{2k}(W_{\wp}),\; k=2, \ldots,2g+1$ of degree at most 3 from the functions $\wp_{1\cdot j,(2k-1)\cdot 1}$.
\end{theorem}

The proof method of Theorem \ref{T-7.1} will be demonstrated on the following examples:
\begin{example}\label{ex-2}
{\rm 1.} Differentiating the relation \eqref{ex-1-1} with respect to $u_1$, we obtain
\begin{equation}\label{ex-2-1}
\wp_{1\cdot 5} = 12\wp_{1\cdot 2}\wp_{1\cdot 3} + 4\wp_{1\cdot 2,3\cdot 1}.
\end{equation}
{\rm 2.} According to formula \eqref{ex-1-1}, we obtain
\begin{equation}\label{ex-2-2}
2\lambda_4 = \wp_{1\cdot 4} - 6\wp_{1\cdot 2}^2 - 4\wp_{1\cdot 1,3\cdot 1}.
\end{equation}
{\rm 3.} According to formula \eqref{ex-1-2}, we obtain
\begin{equation}\label{ex-2-3}
2\wp_{3\cdot 2} = 6(\wp_{1\cdot 2} \wp_{1\cdot 1,3\cdot 1} + \wp_{1\cdot 1,5\cdot 1}) - \wp_{1\cdot 3,3\cdot 1}.
\end{equation}
{\rm 4.} Substituting expressions for $\lambda_4$ (see \eqref{ex-2-2}) and $\wp_{3\cdot 2}$ (see \eqref{ex-2-3})
  into formula \eqref{ex-1-3}, we obtain an expression for the polynomial $\lambda_6$.
\end{example}

The derivation of formulas \eqref{ex-2-1} - \eqref{ex-2-3} and the method of obtaining the polynomial $\lambda_6$ demonstrate
the method of proving Theorem \ref{T-7.1}. Below, this method will be set out in detail in cases $g=1$ (see Section 3.3)
and $g=2$ (see Section 3.4).

\begin{corollary}\label{cor-7-0}
The operator $L$ of differentiation with respect to $u=(u_1,\ldots, u_{2g-1})$ and $\lambda = (\lambda_1,\ldots, \lambda_{4g+2})$
is a derivation of the ring $\mathcal{P}$ if and only if $L\wp_{1\cdot 1,(2k-1)\cdot 1}\in \mathcal{P}$ for $k = 1,\ldots,g$.
\end{corollary}

\begin{proof}
According to part 1 of Theorem \ref{T-7.1}, it suffices to prove that $L\wp_{1\cdot j,(2k-1)\cdot 1}\in \mathcal{P}$ for $j = 2$ and 3,
$k = 1,\ldots,g$. We have $L\wp_{1\cdot j,(2k-1)\cdot 1} = LL_1\wp_{1\cdot (j-1),(2k-1)\cdot 1} = (L_1L + [L,L_1])\wp_{1\cdot (j-1),(2k-1)\cdot 1}$.
Using now that $[L,L_1]\in \mathcal{L}^*$  and the assumption of Theorem \ref{T-7.1}, we complete the proof by induction.
\end{proof}

Set $\mathcal{A} = \mathbb{C}[X]$, where $X = \{ x_{i,2j-1},\; 1 \leqslant i \leqslant 3,\; 1 \leqslant j \leqslant g \}$,
$\deg x_{i,2j-1} = i + 2j - 1$.
\begin{corollary}\label{cor-7-1}
{\rm 1.} The birational isomorphism $J \colon \mathcal{U}_g \to \mathbb{C}^{3g}$ is given by the polynomial isomorphism
\[
J^* \colon \mathcal{A} \longrightarrow \mathcal{P}\; : \; J^*X = W_{\wp}.
\]

{\rm 2.} There is a polynomial map
\[
p \colon \mathbb{C}^{3g} \longrightarrow \mathbb{C}^{2g},\quad p(X) = \lambda,
\]
such that
\[
p^* \lambda_{2k} = \lambda_{2k}(X),\; k=2,\ldots,2g+1,
\]
where $\lambda_{2k}(X)$ are the polynomials from Theorem \ref{T-7.1}, item 2, obtained by substituting \; $W_{\wp}~\longmapsto~X$.

\end{corollary}

The isomorphism $J^*$ defines the Lie $\mathcal{A}$-algebra $\mathcal{L} = \mathcal{L}_g$ with $3g$ generators $L_{2k-1},\; k=1,\ldots,g$
and $L_{2l},\; l=0,\ldots,2g-1$.
In terms of the coordinates $x_{i,2j-1}$, we obtain the following description of the $g$-th stationary flow of KdV system.

\begin{theorem}\label{T-7.2}
{\rm 1.} The commuting operators $L_{2k-1},\; k=1,\ldots,g,$ define on $\mathbb{C}^{3g}$ a polynomial dynamical system
\begin{equation}\label{f-sist}
L_{2k-1}X = G_{2k-1}(X),\; k=1,\ldots,g,
\end{equation}
where $G_{2k-1}(X) = \{ G_{2k-1,i,2j-1}(X) \}$ and $G_{2k-1,i,2j-1}(X)$ is a polynomial that uniquely defines the expression
for the function $\wp_{1\cdot i,(2j-1)\cdot 1,(2k-1)\cdot 1}$ in the form of a polynomial from the functions $\wp_{1\cdot i,(2q-1)\cdot 1}$.

{\rm 2.} System \eqref{f-sist} has $2g$ polynomial integrals $\lambda_{2k} = \lambda_{2k}(X),\; k=2,\ldots,2g+1$.
\end{theorem}

\subsection{\underline{Elliptic functions}}
Consider the curve
\[
V_\lambda = \{ (x,y) \in \mathbb{C}^2\;:\; y^2 = x^3 + \lambda_4 x + \lambda_6 \}.
\]
The discriminant of the family of curves $V_\lambda$ is
\[
\Delta = \{ \lambda = (\lambda_4,\lambda_6)\in \mathbb{C}^2\;:\; 4\lambda_4^3 + 27\lambda_6^2 = 0 \}.
\]
We have the universal bundle $\pi \colon \mathcal{U}_1 \to \mathcal{B}_1 = \mathbb{C}^2\setminus \Delta$
and the mapping
\[
\varphi \colon \mathcal{B}_1\times\mathbb{C} \to \mathcal{U}_1\;:\; \lambda \times\mathbb{C}
\to \mathbb{C}/\Gamma_1(\lambda).
\]

Consider the field $F = F_1$ of functions on $\mathcal{U}_1$ such that the function $\varphi^*(f)$
is meromorphic, and its restriction to the fiber $\mathbb{C}/\Gamma_1(\lambda)$ is an elliptic function
for any point $\lambda \in \mathcal{B}_1$. Using the Weierstrass sigma function $\sigma(u;\lambda)$ for
$\partial = \frac{\partial}{\partial u}$, we obtain
\[
\zeta(u) = \partial\ln \sigma(u;\lambda)\quad \text{and}\quad \wp(u;\lambda) = -\partial\zeta(u;\lambda).
\]

The ring of polynomials $\mathcal{P} = \mathcal{P}_1$ in $F$ is generated by the elliptic functions
$\wp_{1\cdot i},\; i\geqslant 2$. Set $\wp_{1\cdot i} = \wp_{i}$. We have $\wp_{2} = \wp$ and $\wp_{i+1} =
\partial\wp_{i} = \wp_{i}'$. All the algebraic relations between the functions $\wp_{i}$ follow from the relations
\begin{align}
 \wp_{4} &= 6\wp_{2}^2 + 2\lambda_4 \quad (\text{see}\;\eqref{f-3}),\label{f-12}  \\
 \wp_{3}^2 &= 4[\wp_{2}^3 + \lambda_4\wp_{2} + \lambda_6] \quad (\text{see}\;\eqref{f-4}).\label{f-13}
\end{align}
Thus, we obtain a classical result:

\begin{theorem}
{\rm 1.} There is the isomorphism $\mathcal{P}\simeq \mathbb{C}[\wp,\wp',\wp'']$.

{\rm 2.} The projection $\pi \colon \mathcal{U}_1 \to \mathbb{C}^2$ is given by the polynomials
\begin{align}
\label{f-14} & \frac{1}{2}\wp'' - 3\wp^2 = \lambda_4,\\
\label{f-15} & \left( \frac{\wp'}{2} \right)^2 + 2\wp^3 - \frac{1}{2}\wp''\wp = \lambda_6.
\end{align}
\end{theorem}

Consider the linear space $\mathbb{C}^3$ with the graded coordinates $x_2,x_3,x_4,\; \deg x_k=k$.
Set $\mathcal{A}_1 = \mathbb{C}[x_2,x_3,x_4]$.

\begin{corollary}\label{cor-7-7}
The birational isomorphism $J\colon \mathcal{U}_1 \to \mathbb{C}^3$ is given by the ring isomorphism
\[
J^* \colon \mathcal{A}_1 \to \mathcal{P}_1 \;:\; J^*(x_2,x_3,x_4) = (\wp,\wp',\wp'').
\]
\end{corollary}
\begin{proof}
The ring $\mathcal{P}_1$ is generated by elliptic functions $\wp_{i},\; i \geqslant 2,$ where $\wp_{i+1} = \wp_{i}'$.
It follows from formula \eqref{f-14} that $\wp_{5} = 12\wp_{2}\wp_{3}$. Hence, each function $\wp_{i}$ is a polynomial
in $\wp_{2},\,\wp_{3}$ and $\wp_{4}$ for all $i \geqslant 5$.
\end{proof}

\begin{corollary}\label{cor-7-8}
{\rm 1.} The operator $L_1 = \partial$ defines on $\mathbb{C}^3$ a polynomial dynamical system
\begin{equation} \label{F-7-sist}
x_2' = x_3, \quad x_3' = x_4, \quad x_4' = 12x_2x_3.
\end{equation}

{\rm 2.} The system \eqref{F-7-sist} has 2 polynomial integrals
\[
\lambda_4 = \frac{1}{2}x_4 - 3x_2^2 \quad \text{and}\quad \lambda_6 = \frac{1}{4}x_3^2 + 2x_2^3 - \frac{1}{2}x_4x_2.
\]
\end{corollary}

Let us consider the standard Weierstrass model of an elliptic curve
\[
V_g = \{ (x,y) \in \mathbb{C}^2\;:\; y^2 = 4x^3 - g_2x - g_3 \}.
\]
The discriminant of this curve has the form $\Delta(g_2,g_3) = g_2^3 - 27g_3^2$.
We have $V_\lambda = V_g$ where $g_2 = -4\lambda_4$ and $g_3 = -4\lambda_6$.

The elliptic sigma function $\sigma(u;\lambda)$ satisfies the system of equations
\begin{equation} \label{f-16}
\ell_{2i}\,\sigma = H_{2i}\,\sigma, \; i=0,1,
\end{equation}
where
\begin{align*}
  \ell_0 &= 4\lambda_4 \partial_{\lambda_4} + 6\lambda_6 \partial_{\lambda_6};\,\quad H_0 = u\partial - 1;\\
  \ell_2 &= 6\lambda_6 \partial_{\lambda_4} -  \frac{4}{3}\lambda_4^2 \partial_{\lambda_6};\quad H_2 =
\frac{1}{2} \partial^2 + \frac{1}{6} \lambda_4 u^2.
\end{align*}

The operators $\ell_0,\,\ell_2$ and $H_0,\,H_2$, characterizing the sigma function $\sigma(u;g_2,g_3)$
of the curve $V_g$, were constructed in the work of Weierstrass \cite{weier894}. The operators
$L_i \in Der(F_1),\; i=0,1$ and 2, were first found by Frobenius and Stickelberger (see \cite{fs882}).
Below, following work \cite{bl08}, we present the construction of the operators $L_0$ and $L_2$
on the basis of equations \eqref{f-16}.

Let us construct the linear differential operators $\widehat{H}_{2i},\; i=0,1$, of first order,
such that $L_{2i} = \ell_{2i} - \widehat{H}_{2i},\; i=0,1,$ are the differentiations of the ring
$\mathcal{P}_1 = \mathbb{C}[\wp,\wp',\wp'']$.

\vskip .2cm
{\rm 1.} \underline{The formula for $L_0$.}

We have $\ell_0 \sigma = (u\partial - 1)\sigma$. Therefore, $\ell_0\ln \sigma = u\zeta(u)-1$.
Applying the operators $\partial$, $\partial^2$ and using the fact that operators $\partial$,
$\ell_0$ commute, we obtain:
\[
\ell_0\zeta = \zeta - u\wp, \quad \ell_0\wp = 2\wp + u\partial\wp.
\]
Setting $\widehat{H}_0 = u\partial$, we obtain $L_0 = \ell_0 - u\partial$. Consequently
\[
L_0\zeta = \zeta, \quad L_0\wp = 2\wp.
\]

\vskip .2cm
{\rm 2.} \underline{The formula for $L_2$.}

We have $\ell_2 \sigma = \frac{1}{2} \partial^2\sigma - \frac{1}{6} \lambda_4 u^2 \sigma$.
Therefore $\ell_2 \ln \sigma = \frac{1}{2} \frac{\partial^2\sigma}{\sigma} - \frac{1}{6} \lambda_4 u^2$.
We have $\frac{\partial^2\sigma}{\sigma} = -\wp_2 + \zeta^2$. Thus
\begin{equation} \label{F-18}
 \ell_2 \ln \sigma = -\frac{1}{2}\wp_2 + \frac{1}{2}\zeta^2 - \frac{1}{6} \lambda_4 u^2.
\end{equation}
Applying the operators $\partial$ and $\partial^2$ to \eqref{F-18}, we obtain
\[
  \ell_2 \zeta = -\frac{1}{2}\wp_3 + \zeta\partial\zeta - \frac{1}{3}\lambda_4 u, \qquad
  -\ell_2 \wp_2 = -\frac{1}{2}\wp_4 + \wp_2^2 - \zeta\partial \wp_2 - \frac{1}{3}\lambda_4.
\]
Setting $\widehat{H}_2 = \zeta\partial$, we obtain $L_2 = \ell_2 - \zeta\partial$. Consequently,
\[
  L_2 \zeta = -\frac{1}{2}\wp_3 - \frac{1}{3}\lambda_4 u, \qquad
  L_2 \wp_2 = \; \frac{1}{2}\wp_4 - \wp_2^2 + \frac{1}{3}\lambda_4 = \frac{2}{3}\wp_4 - 2\wp_2^2.
\]
Thus, we get the following result:

\begin{theorem}\label{T-7.9}
The Lie $\mathcal{P}_1$-algebra $\mathcal{L}_1$ is generated by operators $L_0,\, L_1$ and $L_2$
such that
\begin{equation}\label{f-19}
  [L_0,\, L_k] = kL_k, \; k=1,2, \qquad
  [L_1,\, L_2] = \wp_2 L_1,
\end{equation}
\begin{equation}\label{f-20}
  L_0 \wp_2 = 2\wp_2;\quad L_1 \wp_2 = \wp_3;\quad L_2 \wp_2 = \frac{2}{3}\wp_4 - 2\wp_2^2.
\end{equation}
\end{theorem}

\begin{proof}
Formulas \eqref{f-19}--\eqref{f-20} completely determine the actions of the operators $L_k, \; k=0,1,2,$
on the ring $\mathcal{P}_1$ by the following inductive formula:
\begin{equation}\label{f-21}
L_k \wp_{i+1} = [L_k,\, L_1]\wp_{i} + L_1 L_k \wp_{i}.
\end{equation}
\end{proof}

\begin{example}
Substituting $k=2$ and $i=2$ in \eqref{f-21}, we obtain
\[
L_2 \wp_3 = [L_2,\, L_1]\wp_{2} + L_1 L_2 \wp_{2} = -5\wp_{2}\wp_{3} + \frac{4}{3}\wp_5.
\]
\end{example}

\subsection{\underline{Hyperelliptic functions of genus $g=2$}}

For each curve with affine part of the form
$$
V_{\lambda} = \left\{ (x,y) \in \mathbb{C}^2\,|\, y^2 = x^5 + \lambda_4 x^3 +
\lambda_6 x^2 + \lambda_8 x + \lambda_{10} \right\},
$$
one can construct a sigma-function $\sigma(u; \lambda)$ (see \cite{bel97}).
This function is an entire function in $u = (u_1, u_3) \in \mathbb{C}^2$ with parameters $\lambda =
(\lambda_4, \lambda_6, \lambda_8, \lambda_{10})\in \mathbb{C}^4$.
It has a series expansion in $u$ over the polynomial ring
$\mathbb{Q}[\lambda_4, \lambda_6, \lambda_8, \lambda_{10}]$ in the vicinity of $0$.
The initial segment of the expansion has the form
\begin{multline}\label{F-22}
\sigma(u; \lambda) = u_3 - {1 \over 3}\, u_1^3 + {1 \over 6}\,
\lambda_6 u_3^3 - {1 \over 12}\, \lambda_4 u_1^4 u_3 - {1 \over 6}\,
\lambda_6 u_1^3 u_3^2 - \\ - {1 \over 6}\, \lambda_8  u_1^2 u_3^3 -
{1 \over 3}\, \lambda_{10} u_1 u_3^4 + \left({ 1 \over 60}\,
\lambda_4 \lambda_8 + {1 \over 120}\, \lambda_6^2\right) u_3^5 +
(u^7).
\end{multline}
Here $(u^k)$ denotes the ideal generated by monomials
$u_1^i u_3^j$, $i+j = k$.

The sigma-function is an odd function in $u$, i.e. $\sigma(-u;\lambda)=-\sigma(u;\lambda)$.

Set
\[
\nabla_{\lambda} = \left(
{\partial \over \partial \lambda_4}, \;   {\partial \over \partial
\lambda_6},\;   {\partial \over \partial \lambda_8}, \;   {\partial
\over \partial \lambda_{10}} \right)\quad \text{and}\quad
\partial_{u_1}={\partial \over \partial u_1}, \;  \partial_{u_3}= {\partial \over \partial u_3}.
\]

We need the following properties of the two-dimensional sigma-function \\ (see \cite{bel997, bl05} for details)\,:

{\bf 1.} The following system of equations holds:
\begin{equation} \label{F-23}
 \ell_i\sigma = H_i\sigma, \quad i = 0,2,4,6,\qquad
\end{equation}
where \;$(\ell_0\;\ell_2\;\ell_4\;\ell_6)^\top=T \, \nabla_\lambda$\,,
\[T =
\begin{pmatrix}
4 \lambda_4 & 6 \lambda_6 & 8 \lambda_8 & 10 \lambda_{10} \\[5pt]
6 \lambda_6 & 8 \lambda_8 - {12 \over 5} \lambda_4^2 & 10 \lambda_{10}
- {8 \over 5} \lambda_4 \lambda_6 & - {4 \over 5} \lambda_4 \lambda_8 \\[5pt]
8 \lambda_8 & 10 \lambda_{10} - {8 \over 5} \lambda_4 \lambda_6 &
4 \lambda_4 \lambda_8 - {12 \over 5} \lambda_6^2 & 6 \lambda_4 \lambda_{10}
- {6 \over 5} \lambda_6 \lambda_8 \\[5pt]
10 \lambda_{10} & - {4 \over 5} \lambda_4 \lambda_8 &
6 \lambda_4 \lambda_{10} - {6 \over 5} \lambda_6 \lambda_8 &
4 \lambda_6 \lambda_{10} - {8 \over 5} \lambda_8^2 \\
\end{pmatrix} \qquad\qquad\quad\quad
\]
and
\begin{align*}
H_0 &= u_1\partial_{u_1}+3u_3\partial_{u_3}-3, \\
H_2 &= {1 \over 2}\,\partial_{u_1}^2 - {4 \over 5}\lambda_4 u_3 \partial_{u_1}+u_1\partial_{u_3} - {3 \over 10}\lambda_4 u_1^2 + {1 \over 10}(15\lambda_8-4\lambda_4^2)u_3^2, \\
H_4 &= \partial_{u_1}\partial_{u_3} - {6\over 5}\,\lambda_6u_3 \partial_{u_1} + \lambda_4 u_3 \partial_{u_3} - {1 \over 5}\,\lambda_6u_1^2 + \lambda_8u_1u_3 + {1 \over 10}(30\lambda_{10} -
6\lambda_6\lambda_4)u_3^2 - \lambda_4, \\
H_6 &= {1 \over 2}\,\partial_{u_3}^2 - {3 \over 5}\lambda_8 u_3 \partial_{u_1} - {1 \over 10}\,\lambda_8u_1^2 + 2\lambda_{10}u_1u_3 - {3 \over 10}\,\lambda_8\lambda_4 u_3^2 - {1 \over 2}\,\lambda_6.
\end{align*}

{\bf 2.} The equation $\ell_0 \, \sigma = H_0 \sigma$ implies that  $\sigma$ is a homogeneous function of degree $-3$ in $u_1$, $u_3$, $\lambda_j$.

{\bf 3.} The discriminant of the hyperelliptic curve $V_\lambda$ of genus 2 is equal to $\Delta= {16 \over 5}\, \det T$.
It is a homogeneous polynomial in $\lambda$ of degree $40$.
Set $\mathcal{B} = \{ \lambda \in \mathbb{C}^4\, :\, \Delta(\lambda) \ne 0 \}$; then the curve $V_\lambda$
 is smooth for $\lambda \in \mathcal{B}$.

 We have
$$
\ell_0\,\Delta = 40 \Delta, \quad \ell_2\,\Delta = 0,\quad \ell_4\,\Delta =
12 \lambda_4 \Delta,\quad \ell_6\,\Delta = 4 \lambda_6 \Delta.
$$
Thus, the fields $\ell_0,\ell_2,\ell_4$ and $\ell_6$ are tangent
to the variety $\{ \lambda\in \mathbb{C}^4\;:\; \Delta(\lambda)=0 \}$.

The present study is based on the following results.

\begin{theorem}[uniqueness conditions for the two-dimensional sigma-function]
\text{ }\\
The entire function $\sigma(u;\lambda)$
is uniquely determined by the system of equations {\rm(\ref{F-23})} and initial condition
$\sigma(u;0)=u_3-\frac{1}{3}u_1^3$.
\end{theorem}

We have the universal bundle $\pi \colon \mathcal{U}_2 \to \mathcal{B}_2 = \mathbb{C}^4\setminus\mathcal{D}$
and the mapping
\[
\varphi \colon \mathcal{B}_2\times\mathbb{C}^2 \to \mathcal{U}_2\;:\; \lambda \times\mathbb{C}^2
\to \mathbb{C}^2/\Gamma_2(\lambda).
\]

Consider the field $F = F_2$ of functions on $\mathcal{U}_2$ such that the function $\varphi^*(f)$
is meromorphic, and its restriction to the fiber $\mathbb{C}^2/\Gamma_2(\lambda)$ is an hyperelliptic function
for any point $\lambda \in \mathcal{B}_2$.

All the algebraic relations between the hyperelliptic functions of genus 2 follow from the relations,
which in our notations have the form:
\begin{align}
\wp_{1\cdot 4} &= 6\wp_{1\cdot 2}^2 + 4\wp_{1\cdot 1,3\cdot 1} + 2\lambda_4,\label{f-26} \\
\wp_{1\cdot 3,3\cdot 1} &= 6\wp_{1\cdot 2} \wp_{1\cdot 1,3\cdot 1} - 2\wp_{3\cdot 2},\label{f-27}
\end{align}
(see \eqref{f-3} for $i=1$ and $i=2$) and
\begin{align}
\wp_{1\cdot 3}^2 &= 4\left[ \wp_{1\cdot 2}^3 + (\wp_{1\cdot 1,3\cdot 1} + \lambda_4)\wp_{1\cdot 2} + \wp_{3\cdot 2}
+ \lambda_6 \right],\label{f-28} \\
\wp_{1\cdot 3}\wp_{1\cdot 2,3\cdot 1} &= 4\wp_{1\cdot 2}^2\wp_{1\cdot 1,3\cdot 1} + 2\wp_{1\cdot 1,3\cdot 1}^2 -
2\wp_{1\cdot 2}^2\wp_{3\cdot 2} + 2\lambda_4\wp_{1\cdot 1,3\cdot 1} + 2\lambda_8, \label{f-29} \\
\wp_{1\cdot 2,3\cdot 1}^2 &= 4(\wp_{1\cdot 2}\wp_{1\cdot 1,3\cdot 1}^2 - \wp_{1\cdot 1,3\cdot 1}\wp_{3\cdot 2} + \lambda_{10}) \label{f-30}
\end{align}
(see \eqref{f-4} for $(i,k) = (1,1),\, (1,2)$ and $(2,2)$).

Consider the linear space $\mathbb{C}^6$ with the graded coordinates $X = (x_2,x_3,x_4),\; Y = (y_4,y_5,y_6)$, $\deg x_k = k,\; \deg y_k = k$.
Set $\mathcal{A}_2 = \mathbb{C}[X,Y]$.

\begin{theorem}\label{T-7.13}
{\rm 1.} The birational isomorphism $J_2 \colon \mathcal{U}_2 \to \mathbb{C}^6$ is given by the isomorphism of polynomial rings
\[
J_2^* \colon \mathcal{A}_2 \longrightarrow \mathcal{P}_2\;:\; J_2^*X = (\wp_{1\cdot 2},\wp_{1\cdot 3},\wp_{1\cdot 4}),\;
J_2^*Y = (\wp_{1\cdot 1,3\cdot 1},\wp_{1\cdot 2,3\cdot 1},\wp_{1\cdot 3,3\cdot 1}).
\]
{\rm 2.} The projection $\pi_2 \colon \mathbb{C}^6 \to \mathbb{C}^4$ is given by the polynomials
\begin{align}
  \lambda_4 &= -3x_2^2 + \frac{1}{2}x_4 - 2y_4, \label{f-31}  \\
  \lambda_6 &= 2x_2^3 + \frac{1}{4}x_3^2 - \frac{1}{2}x_2x_4 - 2x_2y_4 + \frac{1}{2}y_6, \label{f-32} \\
  \lambda_8 &= (4x_2^2 + y_4)y_4 - \frac{1}{2}(x_4y_4 - x_3y_5 + x_2y_6), \label{f-33} \\
  \lambda_{10} &= 2x_2y_4^2 + \frac{1}{4}y_5^2 -\frac{1}{2}y_4y_6.  \label{f-34}
\end{align}
\end{theorem}
\begin{proof}
Using the isomorphism $J_2^*$, we rewrite the relations \eqref{f-26} - \eqref{f-30} in the form
\begin{align}
x_4 &= 6x_2^2 + 4y_4 + 2\lambda_4, \label{f-35} \\
y_6 &= 6x_2y_4 - 2\wp_{3\cdot 2}, \label{f-36} \\
x_3^2 &= 4\left[ x_2^3 + (y_4 + \lambda_4)x_2 + \wp_{3\cdot 2} + \lambda_6 \right], \label{f-37} \\
x_3y_5 &= 2\left[ 2x_2^2y_4 + y_4^2 - x_2^2\wp_{3\cdot 2} + \lambda_4y_4 + \lambda_8 \right], \label{f-38} \\
y_5^2 &= 4\left[ x_2y_4^2 - y_4\wp_{3\cdot 2} + \lambda_{10} \right]. \label{f-39}
\end{align}
Directly from relations \eqref{f-35} - \eqref{f-39}, we obtain the formula for the polynomial mapping $\pi_2$,
that is, the proof of assertion 2 of the theorem.

Set $x_{i+1} = \wp_{1\cdot(i+1)},\; y_{i+3} = \wp_{1\cdot i,3\cdot 1},\; i\geqslant 1$.
Applying the operator $\partial_{u_1}$ to formula \eqref{f-35}, we obtain
\begin{equation}\label{f-40}
  x_5 = 12x_2x_3 + 4y_5 = x_5(X,Y).
\end{equation}
Substituting the expression for $\lambda_4$ from \eqref{f-35} and the expression for $\wp_{3\cdot 2}$ from \eqref{f-36}
into the formula \eqref{f-37} and then applying the operator $\partial_{u_1}$, we obtain
\begin{equation}\label{f-41}
  y_7 = 4x_3y_4 + x_2(x_5 + 4y_5 -12x_2x_3) = y_7(X,Y).
\end{equation}
By induction from formulas \eqref{f-40} and \eqref{f-41}, we obtain the polynomial formulas
\begin{equation}\label{f-42}
x_{i+1} = x_{i+1}(X,Y),\qquad y_{i+3} = y_{i+3}(X,Y).
\end{equation}
From formula \eqref{f-36} we obtain
\begin{equation}\label{f-43}
\wp_{3\cdot 2} = 3x_2y_4 - \frac{1}{2}y_6 = z_{6}(X,Y), \qquad \wp_{3\cdot (i+2)} = \partial_{u_3}^i z_{6}(X,Y) = z_{3i+6}.
\end{equation}
The following formulas complete the proof of assertion 1 of the theorem
\begin{align}
  \partial_{u_3}x_{i+1} &= \partial_{u_1}\wp_{1\cdot i,3\cdot 1} = \partial_{u_1}y_{i+3}(X,Y), \label{f-44} \\
  \partial_{u_3}y_{i+3} &= \partial_{u_3}\wp_{1\cdot i,3\cdot 1} = \wp_{1\cdot i,3\cdot 2} =
  \partial_{u_1}^iz_{6}(X,Y). \label{f-45}
\end{align}
\end{proof}

In the course of the proof of Theorem \ref{T-7.13}, we obtained a detailed proof of Theorem \ref{T-7.1} in the case $g=2$.

Set $L_1=\partial_{u_1}$ and $L_3=\partial_{u_3}$.
We introduce the operators $L_i\in\Der(F_2),\; i=0,2,4,6$, based on the operators $\ell_i-H_i$.

\begin{theorem}\label{T-7.14}
The generators of the $F_2$-module $\Der(F_2)$ are given by the formulas
\[
  L_{2k-1} = \partial_{u_{2k-1}},\; k=1,2, \qquad L_{2k} = \ell_{2k} - \widehat H_{2k},\; k=0,1,2,3,
\]
where
\begin{align*}
  \widehat H_{0} &= u_1\partial_{u_1} + 3u_3\partial_{u_3}, \qquad
  \widehat H_{2} = \left( \zeta_1 - \frac{4}{5}\lambda_4u_3 \right)\partial_{u_1} + u_1\partial_{u_3}, \\
  \widehat H_{4} &= \left( \zeta_3 - \frac{6}{5}\lambda_6u_3 \right)\partial_{u_1} + (\zeta_1 + \lambda_4u_3)\partial_{u_3}, \qquad
  \widehat H_{6} = -\frac{3}{5}\lambda_8u_3\partial_{u_1} + \zeta_3\partial_{u_3}.
\end{align*}
\end{theorem}

\begin{proof}
We will use the methods of \cite{bl08} to obtain the explicit form of operators $L_i$
and to describe their action on the ring $\mathcal{P}_2$.
Note here that this theorem corrects misprints made in \cite{bl08, buch16}.

We have $L_1=\partial_{u_1} \in \Der(F_2)$ and $L_3=\partial_{u_3} \in \Der(F_2)$.

Below we use the fact that $[\partial_{u_k},\ell_q]=0$ for $k=1,3$ and $q=0,2,4,6$.

\vskip .2cm
1). \underline{Derivation of the formula for $L_0$.}

Using \eqref{F-23}, we have $\ell_0\sigma=H_0\sigma=(u_1\partial_{u_1}+3u_3\partial_{u_3}-3)\sigma$. Therefore
\begin{equation}\label{F-26}
\ell_0\ln\sigma = u_1\partial_{u_1}\ln\sigma + 3u_3\partial_{u_3}\ln\sigma - 3.
\end{equation}
Applying the operators $\partial_{u_1}$ and $\partial_{u_3}$ to (\ref{F-26}), we obtain
\begin{align}
\ell_0\zeta_1 &= \zeta_1 - u_1\wp_{1\cdot 2} - 3u_3\wp_{1\cdot 1,3\cdot 1}, \label{F-27} \\
\ell_0\zeta_3 &= 3\zeta_3 - u_1\wp_{1\cdot 1,3\cdot 1} - 3u_3\wp_{3\cdot 2}. \label{F-28}
\end{align}
We apply the operator $\partial_{u_1}$ to (\ref{F-27}) to obtain
\[
-\ell_0\wp_{1\cdot 2} = -2\wp_{1\cdot 2} - u_1\wp_{1\cdot 3} - 3u_3\wp_{1\cdot 2,3\cdot 1}.
\]
Therefore
\[
(\ell_0 - u_1\partial_{u_1} - 3u_3\partial_{u_3})\wp_{1\cdot 2} = 2\wp_{1\cdot 2}.
\]
Applying the operator $\partial_{u_1}$  to (\ref{F-28}), we obtain
\[
-\ell_0\wp_{1\cdot 1,3\cdot 1} = -\wp_{1\cdot 1,3\cdot 1} - u_1\wp_{1\cdot 2,3\cdot 1} - 3\wp_{1\cdot 1,3\cdot 1} - 3u_3\wp_{1\cdot 1,3\cdot 2}.
\]
Therefore,
\[
(\ell_0 - u_1\partial_{u_1} - 3u_3\partial_{u_3})\wp_{1\cdot 1,3\cdot 1} = 4\wp_{1\cdot 1,3\cdot 1}.
\]
Thus, we have proved that
\[
L_0 = \ell_0-u_1\partial_{u_1}-3u_3\partial_{u_3} \in \Der(F_2).
\]

\vskip .2cm
2). \underline{Derivation of the formula for $L_2$.}

Using \eqref{F-23}, we have
\[
\ell_2\sigma = H_2\sigma = \left(\frac{1}{2}\,\partial_{u_1}^2 - \frac{4}{5}\,\lambda_4u_3\partial_{u_1} +
u_1\partial_{u_3} + w_2\right)\sigma
\]
where
\[
w_2 = w_2(u_1,u_3) = -\frac{3}{10}\,\lambda_4u_1^2 + \frac{1}{10}\,(15\lambda_8 - 4\lambda_4^2)u_3^2.
\]
Therefore
\[
\ell_2\ln\sigma = \frac{1}{2}\,\frac{\partial_{u_1}^2\sigma}{\sigma} - \frac{4}{5}\,\lambda_4u_3\partial_{u_1}\ln\sigma  + u_1\partial_{u_3} \ln\sigma + w_2.
\]
It holds that
\[
\frac{\partial_{u_1}^2\sigma}{\sigma} = -\wp_{1\cdot 2,0} + \zeta_1^2.
\]
We get
\begin{equation}\label{F-29}
\ell_2\ln\sigma = -\frac{1}{2}\,\wp_{1\cdot 2} + \frac{1}{2}\,\zeta_1^2 - \frac{4}{5}\,\lambda_4u_3\zeta_1 + u_1\zeta_3 + w_2.
\end{equation}
Applying the operators $\partial_{u_1}$ and $\partial_{u_3}$ to (\ref{F-29}), we obtain
\begin{align*}
\ell_2\zeta_1 &= - \frac{1}{2}\,\wp_{1\cdot 3} - \zeta_1 \wp_{1\cdot 2} + \frac{4}{5}\,\lambda_4u_3\wp_{1\cdot 2} + \zeta_3 - u_1\wp_{1\cdot 1,3\cdot 1} + \partial_{u_1}w_2, \\
\ell_2\zeta_3 &= - \frac{1}{2}\,\wp_{1\cdot 2,3\cdot 1} -  \zeta_1\wp_{1\cdot 1,3\cdot 1} - \frac{4}{5}\,\lambda_4\zeta_1 + \frac{4}{5}\,\lambda_4u_3 \wp_{1\cdot 1,3\cdot 1} -
u_1 \wp_{3\cdot 2} + \partial_{u_3}w_2.
\end{align*}
Applying the operator $\partial_{u_1}$  again, we obtain
\begin{equation*}
-\ell_2\wp_{1\cdot 2} = -\frac{1}{2}\,\wp_{1\cdot 4} + \wp_{1\cdot 2}^2 - \zeta_1\wp_{1\cdot 3} + \frac{4}{5}\,\lambda_4u_3\wp_{1\cdot 3} - 2\wp_{1\cdot 1,3\cdot 1} -
u_1\wp_{1\cdot 2,3\cdot 1} + \partial_{u_1}^2w_2,
\end{equation*}
\begin{multline*}
-\ell_2\wp_{1\cdot 1,3\cdot 1} = -\frac{1}{2}\,\wp_{1\cdot 3,3\cdot 1} + \wp_{1\cdot 2}\wp_{1\cdot 1,3\cdot 1} - \zeta_1\wp_{1\cdot 2,3\cdot 1} + \frac{4}{5}\,\lambda_4\wp_{1\cdot 2} +
\frac{4}{5}\,\lambda_4u_3\wp_{1\cdot 2,3\cdot 1} - \\
- \wp_{3\cdot 2} - u_1\wp_{1\cdot 1,3\cdot 2} + \partial_{u_1}\partial_{u_3}w_2.
\end{multline*}
Thus, we have proved that
\[
L_2=\left(\ell_2-\zeta_1\partial_{u_1} - u_1\partial_{u_3} + \frac{4}{5}\,\lambda_4 u_3\partial_{u_1}\right) \in \Der(F_2).
\]
We have $\partial_{u_1}^2w_2=-\frac{3}{5}\lambda_4$ and $\partial_{u_1}\partial_{u_3}w_2=0$.

\vskip .2cm
3). \underline{Derivation of the formula for $L_4$.}

Using \eqref{F-23}, we have
\[
\ell_4\sigma = H_4\sigma = \left(\partial_{u_1}\partial_{u_3} - \frac{6}{5}\,\lambda_6u_3\partial_{u_1} +
\lambda_4u_3\partial_{u_3} + w_4\right)\sigma
\]
where
\[
w_4 =  -\frac{1}{5}\,\lambda_6u_1^2 + \lambda_8u_1u_3 + \frac{1}{10}\,(30\lambda_{10} - 6\lambda_6\lambda_4)u_3^2 - \lambda_4.
\]
Therefore,
\[
\ell_4\ln\sigma = \frac{\partial_{u_1}\partial_{u_3}\sigma}{\sigma} - \frac{6}{5}\,\lambda_6u_3\partial_{u_1}\ln\sigma  + \lambda_4u_3\partial_{u_3} \ln\sigma + w_4.
\]
It holds that
\[
\frac{\partial_{u_1}\partial_{u_3}\sigma}{\sigma} = - \wp_{1\cdot 1,3\cdot 1} + \zeta_1\zeta_3.
\]
We obtain
\begin{equation}\label{F-30}
\ell_4\ln\sigma = - \wp_{1\cdot 1,3\cdot 1} + \zeta_1\zeta_3 - \frac{6}{5}\,\lambda_6u_3\zeta_1 + \lambda_4u_3\zeta_3 + w_4.
\end{equation}
Applying the operators $\partial_{u_1}$ and $\partial_{u_3}$ to (\ref{F-30}), we obtain
\begin{align*}
\ell_4\zeta_1 &= - \wp_{1\cdot 2,3\cdot 1} -  \wp_{1\cdot 2}\zeta_3 - \zeta_1 \wp_{1\cdot 1,3\cdot 1} + \frac{6}{5}\,\lambda_6u_3\wp_{1\cdot 2} - \lambda_4u_3\wp_{1\cdot 1,3\cdot 1} + \partial_{u_1}w_4, \\
\ell_4\zeta_3 &= - \wp_{1\cdot 1,3\cdot 2} - \wp_{1\cdot 1,3\cdot 1}\zeta_3 - \zeta_1\wp_{3\cdot 2} - \frac{6}{5}\,\lambda_6\zeta_1 + \frac{6}{5}\,\lambda_6u_3 \wp_{1\cdot 1,3\cdot 1} + \lambda_4\zeta_3 -
\lambda_4u_3 \wp_{3\cdot 2} + \partial_{u_3}w_4.
\end{align*}
Applying the operator $\partial_{u_1}$ again, we obtain
\begin{multline*}
-\ell_4\wp_{1\cdot 2} = - \wp_{1\cdot 3,3\cdot 1} - \wp_{1\cdot 3}\zeta_3 + \wp_{1\cdot 2}\wp_{1\cdot 1,3\cdot 1} + \wp_{1\cdot 2}\wp_{1\cdot 1,3\cdot 1} - \zeta_1\wp_{1\cdot 2,3\cdot 1} +
\frac{6}{5}\,\lambda_6u_3\wp_{1\cdot 3} - \\ - \lambda_4 u_3\wp_{1\cdot 2,3\cdot 1} + \partial_{u_1}^2w_4,
\end{multline*}
\begin{multline*}
-\ell_4\wp_{1\cdot 1,3\cdot 1} = - \wp_{1\cdot 2,3\cdot 2} - \wp_{1\cdot 2,3\cdot 1}\zeta_3 + \wp_{1\cdot 1,3\cdot 1}^2 + \wp_{1\cdot 2}\wp_{3\cdot 2} -
\zeta_1\wp_{1\cdot 1,3\cdot 2} + \frac{6}{5}\,\lambda_6\wp_{1\cdot 2} + \\
+ \frac{6}{5}\,\lambda_6u_3\wp_{1\cdot 2,3\cdot 1} - \lambda_4\wp_{1\cdot 1,3\cdot 1} - \lambda_4u_3\wp_{1\cdot 1,3\cdot 2} + \partial_{u_1}\partial_{u_3}w_4.
\end{multline*}
Therefore, we have proved that
\[
L_4 = \left(\ell_4-\zeta_3\partial_{u_1}-\zeta_1\partial_{u_3} + \frac{6}{5}\,\lambda_6 u_3\partial_{u_1} -
\lambda_4u_3\partial_{u_3}\right) \in \Der(F_2).
\]
We have $\partial_{u_1}^2w_4= - \frac{2}{5}\,\lambda_6$ and $\partial_{u_1}\partial_{u_3}w_4=\lambda_8$.

\vskip .2cm
4). \underline{Derivation of the formula for $L_6$.}

Using \eqref{F-23}, we have
\[
\ell_6\sigma = H_6\sigma = \left(\frac{1}{2}\,\partial_{u_3}^2 - \frac{3}{5}\,\lambda_8u_3\partial_{u_1} + w_6\right)\sigma
\]
where
\[
w_6 = -\frac{1}{10}\,\lambda_8u_1^2 + 2\lambda_{10}u_1u_3 - \frac{3}{10}\,\lambda_8 \lambda_4 u_3^2 - \frac{1}{2}\,\lambda_6.
\]
Therefore,
\[
\ell_6\ln\sigma = \frac{1}{2}\,\frac{\partial_{u_3}^2\sigma}{\sigma} - \frac{3}{5}\,\lambda_8u_3\partial_{u_1}\ln\sigma + w_6.
\]
We obtain
\begin{equation}\label{F-31}
\ell_6\ln\sigma = -\frac{1}{2}\,\wp_{3\cdot 2} + \frac{1}{2}\,\zeta_3^2 - \frac{3}{5}\,\lambda_8u_3\zeta_1 + w_6.
\end{equation}
Applying the operators $\partial_{u_1}$ and $\partial_{u_3}$ to (\ref{F-31}), we obtain
\begin{align*}
\ell_6\zeta_1 &= - \frac{1}{2}\,\wp_{1\cdot 1,3\cdot 2} - \zeta_3 \wp_{1\cdot 1,3\cdot 1} + \frac{3}{5}\,\lambda_8u_3\wp_{1\cdot 2} + \partial_{u_1}w_6, \\
\ell_6\zeta_3 &= - \frac{1}{2}\,\wp_{3\cdot 3} -  \zeta_3\wp_{3\cdot 2} - \frac{3}{5}\,\lambda_8\zeta_1 + \frac{3}{5}\,\lambda_8u_3 \wp_{1\cdot 1,3\cdot 1} + \partial_{u_3}w_6.
\end{align*}
Applying the operator $\partial_{u_1}$ again, we obtain
\begin{align*}
-\ell_6\wp_{1\cdot 2} &= -\frac{1}{2}\,\wp_{1\cdot 2,3\cdot 2} + \wp_{1\cdot 1,3\cdot 1}^2 - \zeta_3\wp_{1\cdot 2,3\cdot 1} + \frac{3}{5}\,\lambda_8u_3\wp_{1\cdot 3} + \partial_{u_1}^2w_6,\\
-\ell_6\wp_{1\cdot 1,3\cdot 1} &= -\frac{1}{2}\,\wp_{1\cdot 1,3\cdot 3} + \wp_{1\cdot 1,3\cdot 1}\wp_{3\cdot 2} - \zeta_3\wp_{1\cdot 1,3\cdot 2} + \frac{3}{5}\,\lambda_8\wp_{1\cdot 2} +
\frac{3}{5}\,\lambda_8u_3\wp_{1\cdot 2,3\cdot 1} + \partial_{u_1}\partial_{u_3}w_6.
\end{align*}
Therefore, we have proved that
\[
L_6 = \left(\ell_6-\zeta_3\partial_{u_3} + \frac{3}{5}\,\lambda_8 u_3\partial_{u_1}\right) \in \Der(F_2).
\]
We have $\partial_{u_1}^2w_6= - \frac{1}{5}\,\lambda_8$ and $\partial_{u_1}\partial_{u_3}w_6 = 2\lambda_{10}$.
This completes the proof.
\end{proof}

The description of commutation relations in the differential algebra of Abelian functions of genus $2$
was given in \cite{bl08, buch16}, see also \cite{bel12}.
We obtain this result directly from Theorem \ref{T-7.14} and correct some misprints made in \cite{bl08, buch16}.
To simplify the calculations, we use the following results:

\begin{lemma} \label{lem1}
The following commutation relations hold for $\ell_k$:
\begin{align*}
&[\partial_{u_1}, \ell_k] =0, \quad k = 0, 2, 4, 6,
&
&[\partial_{u_3}, \ell_k] =0, \quad k = 0, 2, 4, 6,
\\
&[\ell_0, \ell_k] = k \ell_k, \quad k = 2, 4, 6,
&
&[\ell_2, \ell_4] = {8 \over 5} \lambda_6 \ell_0 - {8 \over 5} \lambda_4 \ell_2 + 2 \ell_6,
\\
&[\ell_2, \ell_6] = {4 \over 5} \lambda_8 \ell_0 - {4 \over 5} \lambda_4 \ell_4,
&
&[\ell_4, \ell_6] = - 2 \lambda_{10} \ell_0  + {6 \over 5} \lambda_8 \ell_2 - {6 \over 5} \lambda_6  \ell_4 + 2 \lambda_4 \ell_6.
\end{align*}
\end{lemma}
\begin{proof}
This relations follow directly from \eqref{F-23}.
\end{proof}

\begin{lemma} \label{lem2} The operators $L_i$, $i = 0, 1, 2, 3, 4, 6$, act on $- \zeta_1$ and $- \zeta_3$
according to the formulas
\begin{align*}
L_0(- \zeta_1) &= - \zeta_1 , & L_0(- \zeta_3) &= - 3 \zeta_3,\\
L_1(- \zeta_1) &= \wp_{1\cdot 2}, & L_1(- \zeta_3) &= \wp_{1\cdot 1,3\cdot 1},\\
L_2(- \zeta_1) &= {1 \over 2} \wp_{1\cdot 3} - \zeta_3 + {3 \over 5} \lambda_4 u_1, &
L_2(- \zeta_3) &= {1 \over 2} \wp_{1\cdot 2,3\cdot 1} + {4 \over 5} \lambda_4 \zeta_1 + \left({4 \over 5} \lambda_4^2 - 3 \lambda_8\right) u_3,\\
L_3(- \zeta_1) &= \wp_{1\cdot 1,3\cdot 1}, & L_3(- \zeta_3) &= \wp_{3\cdot 2},
\end{align*}
\begin{multline*}
L_4(- \zeta_1) = \wp_{1\cdot 2,3\cdot 1} + {2 \over 5} \lambda_6 u_1 - \lambda_8 u_3, \qquad
L_4(- \zeta_3) = \wp_{1\cdot 1,3\cdot 2} + {6 \over 5} \lambda_6 \zeta_1 - \lambda_4 \zeta_3 - \lambda_8 u_1 +\quad \\
+ 6 \left({1 \over 5} \lambda_4 \lambda_6 - \lambda_{10}\right) u_3,
\end{multline*}
\[
L_6(- \zeta_1) = {1 \over 2} \wp_{1\cdot 1,3\cdot 2} + {1 \over 5} \lambda_8 u_1 - 2 \lambda_{10} u_3, \;\;
L_6(- \zeta_3) = {1 \over 2} \wp_{3\cdot 3} + {3 \over 5} \lambda_8 \zeta_1 - 2 \lambda_{10} u_1 + {3 \over 5} \lambda_4 \lambda_8 u_3.
\]
\end{lemma}

\begin{proof}
For the operators $L_1, L_3$ this result follows from definitions.
For the operators $L_0, L_2, L_4$ and $L_6$ this result follows from  the proof of Theorem \ref{T-7.14}.
\end{proof}

The following theorem completes the description of the action of generators of the Lie $\mathcal{P}_2$-algebra $\mathcal{L}_2$
on the ring of polynomials $\mathcal{P}_2$.

\begin{theorem} \label{T-7.17}
The operators $L_i,\; i=0,1,2,3,4,6$ act on $\wp_{1\cdot 2}$ and $\wp_{1\cdot 1,3\cdot 1}$ according to the formulas
\begin{align*}
  L_0 \wp_{1\cdot 2} &= 2\wp_{1\cdot 2}, \qquad  L_0 \wp_{1\cdot 1,3\cdot 1} = 4\wp_{1\cdot 1,3\cdot 1},\\
  L_1 \wp_{1\cdot 2} &= \wp_{1\cdot 3}, \qquad \;\; L_1 \wp_{1\cdot 1,3\cdot 1} = \wp_{1\cdot 2,3\cdot 1},\\
  L_2 \wp_{1\cdot 2} &= \frac{1}{2}\wp_{1\cdot 4} - \wp_{1\cdot 2}^2 + 2\wp_{1\cdot 1,3\cdot 1} + \frac{3}{5}\lambda_4,\\
  L_2 \wp_{1\cdot 1,3\cdot 1} &= \frac{1}{2}\wp_{1\cdot 3,3\cdot 1} - \wp_{1\cdot 2}\wp_{1\cdot 1,3\cdot 1} - \frac{4}{5}\lambda_4\wp_{1\cdot 2} + \wp_{3\cdot 2},\\
  L_4 \wp_{1\cdot 2} &= \wp_{1\cdot 3,3\cdot 1} - 2\wp_{1\cdot 2}\wp_{1\cdot 1,3\cdot 1} + \frac{2}{5}\lambda_6,\\
  L_4 \wp_{1\cdot 1,3\cdot 1} &= \wp_{1\cdot 2,3\cdot 2} - \wp_{1\cdot 1,3\cdot 1}^2 - \wp_{1\cdot 2}\wp_{3\cdot 2} - \frac{6}{5}\lambda_6\wp_{1\cdot 2} +
\lambda_4\wp_{1\cdot 1,3\cdot 1} - \lambda_8,\\
  L_6 \wp_{1\cdot 2} &= \frac{1}{2}\wp_{1\cdot 2,3\cdot 2} - \wp_{1\cdot 1,3\cdot 1}^2 + \frac{1}{5}\lambda_8,\\
  L_6 \wp_{1\cdot 1,3\cdot 1} &= \frac{1}{2}\wp_{1\cdot 1,3\cdot 3} - \wp_{1\cdot 1,3\cdot 1}\wp_{3\cdot 2} - \frac{3}{5}\lambda_8\wp_{1\cdot 2} - 2\lambda_{10}.
  \end{align*}
\end{theorem}

Note, that in these formulas the parameters $\lambda_{2k},\; k=4,\ldots,10,$ are considered as polynomials $\lambda_{2k}(W_\wp)$ (see \ref{f-31} - \ref{f-34}).

\begin{proof}
See the derivation of the formulas for the operators  $L_{2k},\; k=0,2,3,4,$ in the proof of Theorem \ref{T-7.14}.
\end{proof}

The following result is based on the formulas of Theorem \ref{T-7.14}.

\begin{theorem} \label{t-26}
The commutation relations in the Lie $F_2$-algebra $\Der(F_2)$ of derivations
of the field $F_2$ have the form
\begin{align*}
[L_0, L_k] &= kL_k, \quad k=1,2,3,4,6; &
[L_1, L_2] &= \wp_{1\cdot 2} L_1 - L_3;\\
[L_1, L_3] &= 0; &
[L_1, L_4] &= \wp_{1\cdot 1,3\cdot 1} L_1 + \wp_{1\cdot 2} L_3; \\
[L_1, L_6] &= \wp_{1\cdot 1,3\cdot 1} L_3; &
[L_3, L_2] &= \left(\wp_{1\cdot 1,3\cdot 1} + {4 \over 5} \lambda_4 \right) L_1; \\
[L_3, L_4] &= \left(\wp_{3\cdot 2} + {6 \over 5} \lambda_6 \right) L_1 + \left(\wp_{1\cdot 1,3\cdot 1} - \lambda_4\right) L_3; &
[L_3, L_6] &= {3 \over 5} \lambda_8 L_1 + \wp_{3\cdot 2} L_3; \\
[L_2, L_4] &=\frac{8}{5}\lambda_6 L_0 -\frac{1}{2}\wp_{1\cdot 2,3\cdot 1}L_1 -\frac{8}{5}\lambda_4 L_2 +\frac{1}{2}\wp_{1\cdot 3}L_3 + 2L_6; \hspace{-20mm} & \\
[L_2, L_6] &= \frac{4}{5}\lambda_8 L_0 -\frac{1}{2}\wp_{1\cdot 1,3\cdot 2}L_1 +\frac{1}{2}\wp_{1\cdot 2,3\cdot 1}L_3 -\frac{4}{5}\lambda_4 L_4; & \\
[L_4, L_6] &=-2\lambda_{10} L_0 - \frac{1}{2}\wp_{3\cdot 3}L_1 +\frac{6}{5}\lambda_8 L_2 +\frac{1}{2}\wp_{1\cdot 1,3\cdot 2}L_3 -\frac{6}{5}\lambda_6 L_4 + 2\lambda_4 L_6.
\hspace{-100mm}&
\end{align*}
\end{theorem}

\begin{proof}
Due to linearity, the relation $[L_0, L_k] = k L_k, \quad k=1,2,3,4,6,$
can be checked independently for every summand in the expression for $L_k$.

The expressions for $[L_m, L_n]$, where $m$ or $n$ is equal to $1$ or $3$,
can be obtained by simple calculations using Theorem \ref{T-7.14}.

It remains to prove the commutation relations among $L_2, L_4$ and $L_6$.
We express $[L_m, L_n]$, where $m < n$ and $m, n = 2, 4, 6$, in the form
\[
[L_m,L_n]= a_{m,n,0} L_0 + a_{m,n,-1} L_1 + a_{m,n,-2} L_2 + a_{m,n,-3} L_3 + a_{m,n,-4} L_4 + a_{m,n,-6} L_6.
\]
We have $\deg a_{i,j,-k} = i+j-k$.
Applying both sides of this equation to $\lambda_k$ and using the explicit expressions for $L_k$, we get
\[
[\ell_m, \ell_n] \lambda_k = (a_{m,n,0} \ell_0 + a_{m,n,-2} \ell_2 + a_{m,n,-4} \ell_4 + a_{m,n,-6} \ell_6) \lambda_k.
\]
This formula and Lemma \ref{lem1} yield the values of the coefficients $a_{m,n,-k}$, $k = 0, 2, 4, 6$:
\begin{align}
[L_2, L_4] &=\frac{8}{5}\lambda_6 L_0 + a_{2,4,-1} L_1 -\frac{8}{5}\lambda_4 L_2 + a_{2,4,-3} L_3 + 2L_6; & \label{bbb24} \\
[L_2, L_6] &= \frac{4}{5}\lambda_8 L_0 + a_{2,6,-1} L_1 + a_{2,6,-3} L_3 -\frac{4}{5}\lambda_4 L_4; & \label{bbb26} \\
[L_4, L_6] &=-2\lambda_{10} L_0 + a_{4,6,-1} L_1 +\frac{6}{5}\lambda_8 L_2 + a_{4,6,-3} L_3 -\frac{6}{5}\lambda_6 L_4 + 2\lambda_4 L_6. \label{bbb46}
\end{align}

In subsequent calculations we compare the actions of the left- and right-hand sides of the expressions
\eqref{bbb24}--\eqref{bbb46} on the coordinates $u_1$ and $u_3$.
To this end we use the expressions \eqref{F-23}, Theorem \ref{T-7.14} and Lemma~\ref{lem2}.

We present the calculation of the coefficient $a_{2,4,-1}$. The left-hand side of \eqref{bbb24} gives
\begin{multline*}
[L_2, L_4] u_1 = L_2(-\zeta_3 + \frac{6}{5}\, \lambda_6 u_3) - L_4(-\zeta_1 + \frac{4}{5}\lambda_4u_3) = \\
= L_2(-\zeta_3) + \frac{6}{5} \ell_2(\lambda_6) u_3 - \frac{6}{5} \lambda_6 u_1
- L_4(-\zeta_1) - \frac{4}{5} \ell_4(\lambda_4) u_3 - \frac{4}{5} \lambda_4 (-\zeta_1 - \lambda_4u_3) = \\
= - {1 \over 2} \wp_{1\cdot 2,3\cdot 1} + {8 \over 5} \lambda_4 \zeta_1 - \frac{8}{5} \lambda_6 u_1
+ {2 \over 5} \left(3 \lambda_8 - {16 \over 5} \lambda_4^2\right) u_3.
\end{multline*}
The right-hand side of \eqref{bbb24} gives
\[
[L_2, L_4] u_1 = a_{2,4,-1} + \frac{8}{5}\lambda_4 \zeta_1 - \frac{8}{5}\lambda_6 u_1
+ {2 \over 5} \left( 3 \lambda_8 -\frac{16}{5}\lambda_4^2\right) u_3.
\]
By equating them, we obtain $a_{2,4,-1} = - {1 \over 2} \wp_{1\cdot 2,3\cdot 1}$.

The coefficients $a_{2,4,-3}$, $a_{2,6,-1}$, $a_{2,6,-3}$, $a_{4,6,-1}$ and $a_{4,6,-3}$
are calculated in a similar way.
\end{proof}

\vskip .2cm


\end{document}